\documentclass[submission,copyright,creativecommons]{eptcs}
\providecommand{\event}{TARK 2025} 

\usepackage{amsthm}

\usepackage{iftex}

\ifpdf
  \usepackage{underscore}         
  \usepackage[T1]{fontenc}        
\else
  \usepackage{breakurl}           
\fi

\title{Group Knowledge of Hypothetical Values}
\author{Alexandru Baltag
\institute{Institute for Logic, Language and Computation \\ University of Amsterdam\\
\\
Amsterdam, the Netherlands}
\email{thealexandrubaltag@gmail.com}
\and
Sonja Smets
\institute{Institute for Logic, Language and Computation \\ University of Amsterdam\\
\\
Amsterdam, the Netherlands}
\email{\quad s.j.l.smets@uva.nl}
}
\def\titlerunning{Group Knowledge of Hypothetical Values}
\def\authorrunning{A. Baltag, S. Smets}

\usepackage[T1]{fontenc}
\usepackage{graphicx}
\usepackage{hhline,multirow,tabularx,tikz,xstring,amsmath}

\usepackage{amsfonts,mathrsfs,epic,amsbsy,amstext, amssymb, amscd,amsmath, textcomp}
\usepackage[all]{xy}

\renewcommand{\phi}{\varphi}

\newcommand{\bq}{\begin{quote}}
\newcommand{\eq}{\end{quote}}

\newtheorem{Th}{Theorem}
\newtheorem{ax}{Axiom}
\newtheorem{df}{Definition}
\newtheorem{pr}{Proposition}
\newtheorem{cl}{Corollary}
\newtheorem{re}{Remark}
\newtheorem{as}{Assumption}
\newtheorem{wg}{Wild Guess}

\newtheorem{lemma}{Lemma}

\newcommand{\Agents}{\mathcal{A}}

\newcommand{\ot}{\leftarrow}

\newcommand{\bM}{{\mathbf M}}

\newcommand{\bS}{{\mathbf S}}

\newcommand{\myComment}[1]{}

\newcommand{\ux}{\overline{{x}}}
\newcommand{\uy}{\overline{{y}}}
\newcommand{\uz}{\overline{{z}}}

\newcommand{\bth}{\begin{Th}}
\newcommand{\Eth}{\end{Th}}
\newcommand{\bax}{\begin{ax}}
\newcommand{\eax}{\end{ax}}
\newcommand{\blm}{\begin{lm}}
\newcommand{\elm}{\end{lm}}
\newcommand{\bdf}{\begin{df}}
\newcommand{\edf}{\end{df}}
\newcommand{\bpr}{\begin{pr}}
\newcommand{\epr}{\end{pr}}
\newcommand{\bcl}{\begin{cl}}
\newcommand{\ecl}{\end{cl}}
\newcommand{\bre}{\begin{re}}
\newcommand{\ere}{\end{re}}
\newcommand{\bas}{\begin{as}}
\newcommand{\eas}{\end{as}}
\newcommand{\bwg}{\begin{wg}}
\newcommand{\ewg}{\end{wg}}
\newcommand{\bex}{\begin{ex}}
\newcommand{\eex}{\end{ex}}

\newcommand{\bit}{\begin{itemize}}
\newcommand{\eit}{\end{itemize}\par\noindent} \newcommand{\beq}{\begin{equation}}
\newcommand{\eeq}{\end{equation}\par\noindent} \newcommand{\beqa}{\begin{eqnarray*}}
\newcommand{\eeqa}{\end{eqnarray*}\par\noindent} \newcommand{\beqn}{\begin{eqnarray}}
\newcommand{\eeqn}{\end{eqnarray}\par\noindent}

\newcommand{\sima}{\sim_a}

\newcommand{\simB}{\sim_B}

\newcommand{\simA}{\sim_A}

\newcommand{\bD}{\mathbf{D}}

\newcommand{\gA}{\mathfrak A}

\begin{document}
\maketitle

\begin{abstract}
In recent years, epistemic logics have been extended with operators $K_a x$ for knowledge of (the value of) a variable $x$ (by an agent $a$). We study dynamic versions of these logics, enriched with modalities for semi-public data-exchange events (e.g., public announcements, data-sharing within a subgroup, or changing the value of a variable).
To obtain a complete axiomatization of data-exchange events, in the presence of equality $x=y$ and $K_ax$, one needs to extend the logic further: first, with an operator for distributed knowledge $K_A x$ of the value (by a \emph{group} of agents $A$); next, with a conditional version of this: distributed knowledge $K_A^\varphi x$ (of the value by a group) given some hypothetical condition (expressed by some proposition $\varphi$); then, with definite descriptions $x_A^\varphi$, denoting the `hypothetical' value of $x$ according to $A$'s (distributed) knowledge given condition $\varphi$. In order to deal with common knowledge in the presence of semi-public data exchanges, we also need to add a \emph{novel conditional version} of the recent concept of \emph{common distributed knowledge}. We investigate the resulting logic, giving examples and presenting a complete axiomatization and a decidability proof.
\end{abstract}
%
%

%
\section{Introduction}

\vspace{-2mm}

Information may come in both propositional and non-propositional (e.g. numerical or graphic) form. Such \emph{data} can be encoded as \emph{values of local variables} (e.g. a cryptographic key or a password), that can be stored in certain \emph{locations} or `sources'
(e.g., databases, websites, etc.). Each source can be associated with an \emph{agent}, either because it actually is the knowledge base of a (natural or artificial) agent, or because we think of the source itself as an abstract `agent' (possessing exactly the
information that is stored at it). More complex data may have an \emph{extended location}, being \emph{distributed} among a number of agents.\footnote{One can recover the value of such complex variables only by accessing several sources.}

In this setting, we will consider a number of \emph{data-exchange actions} or `events'.
The agents may \emph{jointly receive some propositional information}, in the form of a \emph{public announcement}.
An agent may \emph{change the values of her local variables}, according to some substitution function. Or she may \emph{gain access to another site}, in which case she can be assumed to instantly `read' (and copy in her own knowledge base) all the
information stored at that source. An agent may in fact gain simultaneous access to multiple such sources, and at the same time different agents may gain access to various sources. In this way, an agent may store data that are not `owned' by her (e.g. values of some other agents' local variables)\footnote{But we assume that each agent can change only the values of her own local variables.}. Groups of agents may \emph{share or exchange data} in a symmetric or asymmetric manner. Finally, we can also consider events that combine all these features. In general, such data-exchange events can be \emph{public} (when it is common
knowledge that the data stored at a given location has become visible to
everybody, cf. the WikiLeaks case), or \emph{semi-public} (when the data are accessible only
to some agents, but this fact is common knowledge), or \emph{semi-private} (when
some databases are publicly accessed, others are privately accessed, and these
facts are known to some but not to all of the agents), or even \emph{fully private}
(e.g. a \emph{secret hacking} of a database by an unsuspected intruder).

However, in this extended abstract, we follow the approach in \cite{BS20}, by
\emph{restricting ourselves to `semi-public' data-exchanges}. The reason is that the syntax, the axioms and the proofs are simpler in this case, allowing us to focus at the new sources of complexity posed by the introduction of explicit variables and of the new operator of (conditional, distributed) knowledge of the value of a variable.
The generalization to arbitrary events is left to the full journal version of this paper.

\smallskip

To reason about these scenarios, we introduce a logic with data variables $x$ and equality of values $x=y$ (as well as other predicates and functions of variables), endowed with a Kripke-based semantics that can model both the traditional account of `knowing that' ($K_a \varphi$ for \emph{agent $a$ knows that $\varphi$}) and an account of `knowing what' ($K_a x$ for \emph{$a$ knows the value of $x$}). While work on the former (`propositional') knowledge $K_a \varphi$ is well known in standard epistemic logic, work on the later form $K_a x$ traces back to Plaza \cite{Plaza} and subsequent investigations by \cite{vEGWang,Yanjing,WangFan1,WangFan2,GuWang,Ding,Hong,Baltag2016} on formalizing `knowledge de re'. While Plaza axiomatized the static logic of `knowing what', the extension with public announcement operators $[!\varphi]\psi$ was only later axiomatized by Wang and Fan \cite{WangFan2} via the introduction of a \emph{conditional or `hypothetical' version of `knowing what'}: $K_a^\varphi x$ means that $a$ knows the value of $x$ given the information that proposition $\varphi$ was the case. Note that this type of \emph{conditional knowledge may not imply the knowledge of the `real' value of $x$} (in the actual world), but only the fact that the agent \emph{can uniquely determine a hypothetical value} of $x$ (that would apply only to worlds in which $\varphi$ was the case).

Part of the novelty of our setting is that we have a \emph{way to explicitly refer to such hypothetical values}, via the introduction of a new term $x_a^\varphi$, denoting the \emph{(unique) value that $x$ would have according to agent $a$ if $\varphi$ were true}. (Here, we take $x_a^\varphi$ to be undefined if such a unique value doesn't exist, either because the agent $a$ knows that $\varphi$ is false, or because the hypothesis $\varphi$ is not enough for $a$ to determine a unique such value of $x$.)
These complex terms denoting hypothetical values seem to be required by the presence of equality $x=y$ and other predicates: e.g. we need a way to express the fact that the current value of $x$ (at a world $s$ at which $\varphi$ is false) is the same as the hypothetical value of $y$ (at worlds indistinguishable from $s$ in which $\varphi$ were true); and this can be expressed by writing that $x=y_a^\varphi$ holds at $s$.

The presence of data-exchange events forces us to further extend the static base of our logic. As noted in \cite{BS20}, when an agent $a$ gains access to another agent $b$'s knowledge base, her new state of knowledge matches the \emph{distributed knowledge} of the group $\{a, b\}$. So on the side of propositional knowledge, we need operators $K_A \varphi$ for \emph{distributed knowledge within group $A$}.\footnote{The notation $D_A\varphi$ is typically used for distributed knowledge, but we prefer here $K_A\varphi$ because we think of this as capturing a natural notion of (virtual) group knowledge. Furthermore, this will allow us to use the notation $C_{\gA}\varphi$ for ``common distributed knowledge", playing on the analogy between the relationship of individual knowledge $K_a$ with common knowledge $C_A$ on the one hand, and the relationship of distributed knowledge $K_A$ with common distributed knowledge $C_\gA$ on the other hand.} A similar move is also necessary on the side of ``knowing what'': one needs operators $K_A x$ that express that \emph{the group $A$ has distributed knowledge of the value of $x$}. And, because of the simultaneous presence of public announcements, we need in fact a \emph{conditional version} of these operators: $K_A^\varphi x$ stands for group $A$ having a \emph{conditional distributed knowledge value of $x$, given the hypothesis $\varphi$.} Furthermore, due again to the presence of equality and other predicates, we are lead to generalize the terms for hypothetical-values to groups: $x_A^\varphi$ will denote the \emph{(unique) value that $x$ would have according to (the distributed knowledge of) group $A$ if $\varphi$ were true}.\footnote{In addition to such terms for hypothetical values $x_A^\varphi$, we also have other complex terms of more traditional kinds, formed by applying functions $F(x_1, \ldots, x_n)$, or defined by cases.}

We similarly follow \cite{BS20} into generalizing common knowledge $C_A \varphi$ to super-groups $\gA$: for a family $\gA$ of groups of agents, $C_\gA \varphi$ denotes \emph{common distributed knowledge} of $\varphi$, equivalent to the infinite conjunction of all finite iterations of $K_A K_B\ldots \varphi$ for all $A,B, \ldots\in \gA$: every group $A\in \gA$ has distributed knowledge that $\varphi$, and also has distributed knowledge that every other group $B\in \gA$ has distributed knowledge that $\varphi$, etc. Once again, this move is necessitated by the presence of (semi-public) data-exchange events, but the notion of common distributed knowledge has independent motivation and applications \cite{BS20}.
In fact, a further extension is needed here, in presence of public announcements: the novel notion of \emph{conditional common distributed knowledge} $C_\gA^\theta\varphi$ appears for the first time in our current submission.\footnote{We stopped here, due to the limited size of this extended abstract. But the generalization to arbitrary events would require a further, wilder extension of the notion of common (distributed) knowledge, namely the introduction of polyadic modalities as in \cite{BS24}, denoting ``conditional levels of group knowledge'' in the spirit of \cite{Parikh}.}

\smallskip

In this paper we study the resulting logic of semi-public data-exchange events in the presence of (conditional) group knowledge of (hypothetical) values of variables, (conditional) common distributed knowledge and the other ingredients mentioned above. The main results are a complete axiomatization and a proof of decidability for this Dynamic Logic of (group) Knowledge of (hypothetical) values $DLKV$.

\vspace{-4mm}

\section{Syntax and Semantics}

\vspace{-2mm}

A \emph{vocabulary} $\mathcal{V}=(V, \Agents, ag, Pred, Funct, ar)$ consists of the following:
\begin{itemize}
\item a (finite or infinite) set $\Agents$ of \emph{agents} (or `locations', or databases) $a, b, \ldots$;
\item a (finite or infinite) set $C$ of \emph{constants}, that include \emph{two logical constants} $0, 1$, and \emph{one `improper' constant} $\uparrow$ (denoting the `\emph{undefined}' value).
\item a (finite or infinite) set $V$ of \emph{(basic) local variables} $v_a^1$, $v_a^2, \ldots$, each being `located' at some location $a\in\Agents$ (or `owned' by some agent $a$) where its correct value is always stored; we indicate the location/owner via the subscript in $v_a$, but we sometimes drop the subscript, writing just $v$ (or $v^1, v^2, \ldots$), when we are not interested in the variable's location;
\item a set $Pred$ of \emph{predicate symbols} $P, Q, \ldots$, including the  \emph{equality predicate} ($=$), together with an \emph{arity} $ar: Pred\to \mathbb{N}$, mapping each symbol $P\in Pred$ to a number $ar(P)\in \mathbb{N}$, where $ar(=)=2$
    \item and a set $Funct$ of \emph{function symbols} $F, G, \ldots$, also endowed with an \emph{arity map} $ar: Funct\to \mathbb{N}$, associating to each function symbol $F\in Funct$ a natural number $ar(F)\in \mathbb{N}$.
\end{itemize}



\par\noindent\textbf{Groups and Super-groups} Given the vocabulary $\mathcal{V}$, a \emph{group} $A\subseteq\Agents$ is a finite, non-empty set of agents in $\Agents$. We use letters $A, B, \ldots$ to denote groups. A \emph{supergroup} $\mathfrak{A} \subseteq \mathcal{P}(\Agents)$ is a finite, non-empty set of groups. We use letters $\mathfrak{A}, \mathfrak{B},\ldots$ to denote supergroups.

\smallskip
\par\noindent\textbf{Syntax of $DLKV$}. The \emph{Dynamic Logic of (group) Knowledge of (hypothetical) Values} $DLKV(\mathcal{V})$ (over the vocabulary $\mathcal{V}$) has a syntax, consisting of: a set $Var:=Var (\mathcal{V})$ of compound \emph{variables} (or
`terms') $x$ (over the given vocabulary $\mathcal{V}$); a set $Fml;=Fml(\mathcal{V})$ of \textit{formulas} $\varphi$ (over $\mathcal{V}$);
and a set $\mathcal{E}$ of (`semi-public') \emph{events} $e$;
these sets are simultaneously defined by mutual recursion:
$$
\begin{array}{cc cc cc cc cc cc cc}
x & ::= &  c &|& v &|&
x|_\varphi y &|&
F(\overline{x}) &|& x_A^\varphi
&|&
e(x)
\\
\vspace{0.2cm}
\varphi & :: =  & P\overline{x}  &|&  \neg\varphi &|& \varphi \wedge \varphi         &|&
K_A \varphi
&|&
C_{\mathfrak{A}}^\varphi \varphi
&|&
[e] \varphi
\\
e & ::= & ! \Phi /\sigma && && && && &&
\end{array}
$$
where: $c\in C$ is any constant symbol; $v\in V$ is any basic variable; $x\in Var$ is any variable (term), and
$\overline{x}=(x_1, \ldots, x_n)$ are $n$-tuples of variables in $Var$; $P\in Pred$ is any predicate symbol of arity $n$; $F$ is any $n$-ary function symbols\footnote{In particular, recall that the predicate symbols include the equality predicate $=$, and the constants $0,1,\uparrow$. As usual, we will write $x=y$ instead of $=xy$.}; $a\in \Agents$ is any agent, $A\subseteq \Agents$ is any group of agents, and $\mathfrak{A}\subseteq\mathcal{P}(\Agents)$ is any supergroup; $\Phi\subseteq Fml$ is a finite set of (already formed) formulas;
and finally $\sigma: V\cup \Agents\to Var\cup \mathcal{P}(\Agents)$ is a `substitution' map, assigning to each basic variable $v\in V$ some term $\sigma(v)\in Var$ and assigning to each agent $a\in \Agents$ some group of agents $\sigma(a)\subseteq \Agents$. These components are subject to the following two constraints, for all $a\in \Agents$ and all basic local variables $v_a$ of agent $a$:
\vspace{-1mm}
\begin{enumerate}
\item $a\in \sigma(a)$ (\emph{each agent continuously accesses her own database});
\item $K_a \sigma(v_a)\in \Phi$ (\emph{agents know what values they reassign to their own local variables}).
\end{enumerate}
\par\noindent\textbf{The static fragment $LKV$}. The \emph{Logic of (group) Knowledge of (hypothetical) Variables} $LKV(\mathcal{V})$ is the `static fragment' of $DLKV$, obtained by removing the dynamic operators $[e]\varphi$ and $e(x)$ from (all the three components of) the recursive definition of the language $DLKV$.


\smallskip
\par
\noindent\textbf{Precondition, Access and Postconditions of an event} Furthermore, for any event
$e=!\Phi/\sigma$, any basic variable $v\in V$ and any agent $a\in\Agents$, we use the notations

\vspace{0.1cm}
\par\noindent
{\centerline{$pre_e \, \, :=\,\, \bigwedge \Phi, \,\,\,\,\, \,\,\,\,\, e(a)\,\, :=\,\, \sigma(a),
\,\,\,\,\, \,\,\,\,\, post_e (v) \, \, :=\,\, \sigma(v)$}}
\par\noindent
The first of these is called the \emph{precondition} of event $e$ (intuitively characterizing the conditions of possibility of this event);
the second is called the \emph{access map} of event $e$ (intuitively giving for each agent $a\in \Agents$ the list $e(a)\subseteq \Agents$ of agents or locations whose information is being accessed by agent $a$ during this event);
while the third is the \emph{postcondition map} of event $e$, assigning to each variable $v\in V$ the term $\sigma (v)\in Var$ (and intuitively capturing the variable-changes induced by the event).

\smallskip
\par
\noindent\textbf{Informal meaning of our language} 
Our variables are meant to express `data' (which can be numerical, propositional, graphical, etc), and will take various values in various states (or else they can be undefined, in which case they will be assigned the `improper' value $\uparrow$, for `undefined'). The basic variables are either constants $c$ (taking the same value in all states, and keeping their value unchanged) or local variables $v_a$, each `owned' by some agent $a$. From these, we recursively form more complex variables by: \emph{applying function} symbols $F(x_1, \ldots, x_n)$; \emph{defining new terms by cases} ($x|_\varphi y$ is a variable taking the same value as $x$ if $\varphi$ is true, and otherwise taking the same value as $y$); forming \emph{definite descriptions} $x_A^\varphi$ to denote the hypothetical value that $x$ would take according to group $A$ if $\varphi$ were true; and \emph{applying the (substitution involved in) event $e$ to any variable $x$:} $e(x)$ takes the value that $x$ will acquire after event $e$

At the propositional side, we have: atoms $P\ux$ describing relations between the values of a tuple $\ux=(x_1, \ldots, x_n)$ of variables; Boolean operators; group (distributed) knowledge $K_A\varphi$; conditional common distributed knowledge $C_\gA^\theta \varphi$ (pf $varphi$, among the groups in a supergroup $\gA$, conditional on some hypothesis $\theta$); and dynamic modalities $[e]\varphi$, saying as usually that $\varphi$ will hold after event $e$.

Finally, the expressions $e:=!\Phi/\sigma$ for semi-public data-exchange events are of a particular simple kind.
Our events have no epistemic structure (-no epistemic relations as in in standard Dynamic Epistemic Logic), since
the structure of a `semi-public' event is fully transparent: \emph{what is happening} (-the presuppositions of this action, who gains access to whose database and which variables are changed and how) \emph{is common knowledge}; although \emph{the content of this happening} (-the actual data that are being accessed, the current values of variables, as well as the next values acquired after the event) \emph{may remain private}. So event expressions are simple pairs of a set $\Phi$ of precondition formulas and a substitution $\sigma$. The set $\Phi$ comprises all presuppositions or `preconditions' of event $e$: the propositions that must be true in order for the event to happen. The precondition formula $pre_e$ collects them into one proposition, capturing in a sense the propositional information carried by event $e$. The substitution $\sigma$ captures the event's \emph{effect}: for any agent $a$, $e(a):=\sigma(a)$ is the list of all the sources accessed by $a$ during this event; while for any local variable $v$, $post_e(v):=\sigma(v)$ denotes the value that $v$ acquires after the event.


\smallskip
\par
\noindent\textbf{Extending Access to Groups and Supergroups} For any event $e\in \mathcal{E}$, we can extend its access map from single agents to groups $A$ and supergroups ${\mathfrak A}$, by putting:
\vspace{0.1cm}
\par\noindent
{\centerline{$e(A) \, :=\, \bigcup_{a\in A} e(a), \,\,\,\,\,\,\,\,\,\,\,\, \, e[\mathfrak{A}] \, :=\, \{e(A): A\in \mathfrak{A}\}$}

\smallskip

\par
\noindent\textbf{More Abbreviations}. We define as usual the other Boolean connectives. We also use the abbreviations, for variables $x\in Var$ and tuples $\ux=(x_1, \ldots, x_n)\in Var^n$:
\vspace{-2mm}
$$\top \, :=\, (\uparrow=\uparrow),  \,\,\,\,\, \,\,\,\,\, \bot  \, :=\, \neg\top,  \,\,\,\,\, \,\,\,\,\, ?_\varphi \, :=\,
1|_\varphi 0,  \,\,\,\,\, \,\,\,\,\, x\!\!\uparrow \, :=\, (x=\uparrow),
 \,\,\,\,\, \,\,\,\,\,
x\!\!\downarrow \, :=\, \neg x\!\!\uparrow,  \,\,\,\,\, \,\,\,\,\,
\ux\!\!\downarrow \, :=\, \bigwedge_{i=1}^{n} x_i\!\!\downarrow ,
\vspace{-4mm}$$
$$K_A^\theta \varphi \,\, :=\,\, K_A (\theta\to \varphi), \,\,\,\,\, \,\,\,\,\,   K_A^\theta x  \,\, :=\,\, K_A^\theta (x= x_A^\theta), \,\,\,\,\, \,\,\,\,\, K_A^\theta \ux \,\, :=\,\, \bigwedge_{i=1}^n K_A^\theta x_i,  \,\,\,\,\, \,\,\,\,\,
K_A x \,\, :=\,\, K_A^\top x,   \,\,\,\,\, \,\,\,\,\, K_A \ux \,\, :=\,\, K_A^\top \ux, \vspace{-3mm}$$
$$\langle K_A\rangle \varphi  \,\, :=\,\ \neg K_A\neg \varphi,
\,\,\,\,\, \,\,\,\,\, \langle K_A^\theta\rangle \varphi  \,\, :=\,\ \neg K_A^\theta\neg \varphi,
\,\,\,\,\, \,\,\,\,\, \langle C_{\gA}^\theta\rangle \varphi  \,\, :=\,\ \neg C_{\gA}^\theta \neg \varphi,
\,\,\,\,\, \,\,\,\,\, \langle e\rangle \varphi  \,\, :=\,\ \neg [e] \neg \varphi,$$
$$C_A^\theta \varphi \,\, :=\,\, C_{\{\{a\}: a\in A\}}^\theta \varphi, \,\,\,\,\, \,\,\,\,\,
C_A \varphi \,\, :=\,\, C_A^\top \varphi, \,\,\,\,\, \,\,\,\,\,
C_\gA x  \,\, :=\,\, C_\gA (\bigwedge_{A\in \gA} K_A x),
\,\,\,\,\, \,\,\,\,\, C_A x  \,\, :=\,\, C_{\{\{a\}: a\in A\}} x  \, =\, C_A (\bigwedge_{a\in A} K_a x).\vspace{-2mm}$$
So \emph{`knowledge what' is definable in our logic, in all its forms} (individual, common, distributed and common distributed knowledge of the value of a variable), via the above abbreviations defining $K_a x$, $K_A x$, $C_A x$ and $C_\gA x$. Also, note that, for the empty tuple $\lambda=()$, we have $K_A^\theta \lambda = \bigwedge \emptyset =\top$ by definition; while the Boolean variable $?_\varphi$ takes value $1$ if $\varphi$ is true, and value $0$ otherwise.

\smallskip

\par
\noindent\textbf{Abbreviations for events} Whenever some component of an event is `trivial', we may skip it. E.g. if the set of preconditions $\Phi=\emptyset$ in $e=!\Phi/\sigma$ (so $pre_e=\top$), we can write $e=!\sigma$; if in addition the access map is trivial, i.e. we have $\sigma(a)=\{a\}$ for all agents $a\in \Agents$, then we can just write the (non-trivial) post-conditions (only for those variables $v$ for which $\sigma(v)$ is not the same term as $v$): e.g. writing $!(v:= x)$ for a semi-public substitution event which only affects variable $x$; and dually, if both $\Phi=\emptyset$ and the post-conditions are trivial (leaving all variables the same), we can write just the (non-trivial part of the) access map, e.g. writing $!(a:b)$ for a semi-public event in which the only thing happening is that agent $a$ gains access to agent $b$'s knowledge base (so $\sigma(a)=\{a, b\}$ and $\sigma(i)=\{i\}$ for all other agents $i\neq a$).

\smallskip

\par
\noindent\textbf{Extended location} For any \emph{static expression} (i.e., term or formula $\alpha\in Var\cup Fml$ in the static language $LKV$), the set $\mathcal{A}(\alpha)$ of \emph{agents involved in $\alpha$} can be thought of as the ``\emph{extended location}" of that expression: its value (or truth value) is distributed among the agents in $\mathcal{A}(\alpha)$. Formally, we put:
\vspace{-2mm}
$$\mathcal{A}(v_a)\,\, =\,\, \{a\}, \,\,\,\,\,  \mathcal{A}(c)=\emptyset,  \,\,\,\,\, \mathcal{A}(\varphi\wedge \psi)  \,\, =\,\, \mathcal{A}(\varphi)\cup \mathcal{A}(\psi), \,\,\,\,\, \mathcal{A}(\neg\varphi)=\mathcal{A}(\varphi),
\vspace{-2mm}$$
$$\mathcal{A}(x_A^\varphi)= \mathcal{A}(K_A \varphi)  \,\, =\,\, A,  \,\,\,\,\,  \,\,\,
\mathcal{A}(F(\overline{x}))=\mathcal{A}(P\overline{x})\,\, =\,\, \bigcup_{i=1,n} \mathcal{A}(x_i) \,\,\,\,\, \mbox{ for $\overline{x}=(x_i)_{1\leq i\leq n}$}.
\vspace{-4mm}$$
\par\noindent\textbf{Semantics: epistemic data models} We assume given a \emph{first-order model} $\bD=(D, I)$, consisting of: a \emph{domain of values (objects)} $D$, that includes an object $\uparrow_\bD\in D$ (the \emph{`undefined' value}); and an \emph{interpretation function} $I$, mapping each constant symbol $c\in C$ into some value $c_\bD:=I(c)\in D$, each functional symbol $f$ of arity $n$ into a function $I(f): D^n\to D$, and each
relational symbol $P$ of arity $n$ into a set $I(P)\subseteq D^n$ of $n$-tuples of values in $D$, with \emph{the equality symbol $=$ being interpreted as the identity relation $I(=)=\{(d,d):d\in D\}$ on the set $D$}.

An \emph{epistemic data model over $\bD$} is a tuple $\bM= (S, \sim, \bullet(\bullet))$, where:
$S$ is a set of \emph{states} (or `possible worlds'), typically denoted by $s, w, \ldots$; $\sim: \Agents\to \mathcal{P}(S\times S)$ maps each agent $a\in\Agents$ to some \emph{equivalence relation} $\sima\subseteq W\times W$,
called \emph{agent $a$'s `indistinguishability' relation};
and $\bullet(\bullet):S\times V \to D$ maps states $s$ and basic variables $v\in V$ into \emph{values} $s(v)\in D$; where we require that \emph{agents know their own data}, i.e.:
\vspace{-2mm}
$$s\sim_a w \, \mbox{ implies } s(v_a)=w(v_a).$$
\par\noindent\textbf{Group indistinguishability} For every group $A\subseteq \Agents$, we introduce \emph{group indistinguishability relations} $\simA :=\bigcap_{a\in A} \sima$ by \emph{taking intersections}, i.e.: $s \simA w$ iff $s\sima w$ for all $a\in A$.

\medskip
\par\noindent\textbf{Satisfaction, Term-Interpretation, Update} We simultaneously define three  semantic notions: (1) a \emph{satisfaction relation} $s\models_{\bM} \varphi$ between states $s\in S$ (in a data model $\bM$) and formulas $\varphi\in Fml$; (2) an \emph{extended value function} mapping $(state, variable)$-pairs $(s,x)\in S\times Var$
into values $s(x)_\bM\in D$; and (3) an \emph{update function} mapping data models $\bM= (S, \sim, \bullet(\bullet))$ and events $e$ into \emph{updated models} $e(\bM)=(S^e, \sim^e, \bullet(\bullet)^e)$ (over the same first-order model $\bD$).
Whenever the model is understood, we skip the subscript $\bM$, writing simply $s\models \varphi$ and $s(x)$. The definition is by mutual recursion:
\[
\begin{array}{llll}
s\models_\bM P x_1\ldots x_n \  \ & \mbox{ iff } \ \ & (s(x_1)_\bM, \ldots, s(x_n)_\bM)\in I(P)\\
s\models_\bM\neg\varphi  \  \ & \mbox{ iff } \ \ & s\not\models_\bM\varphi \\
s\models_\bM\varphi\wedge\psi \  \ & \mbox{ iff } \ \ &  w\models_\bM\varphi \mbox{ and } w\models_\bM\psi
\\
s\models_\bM K_A \varphi \ \ & \mbox{ iff } \ \ & w\models_\bM\varphi \mbox{ for all $w\simA s$}
\\
s\models_\bM C_{\gA}^\theta\varphi  \ \ & \mbox{ iff } \ \ & w\models_\bM\varphi \mbox{ for every chain $s=s_0\sim_{A^1} s_1 \sim_{A^2} \cdots \sim_{A^n} s_n =w$}
\\
{} \ \ & {}  \ \ & \mbox{ (of any finite length $n\geq 0$) with $A_k\in {\gA}$ and $s_k\models_\bM\theta$ for all $k\geq 1$}
\\
s\models_\bM [e] \varphi \ \ & \mbox{ iff } \ \ & s\models_\bM pre_e \mbox{ implies } s\models_{e(\bM)} \varphi
\\
\\
s(c)_\bM  \ \ & = \ \ & c_\bD \mbox{ as given by the interpretation $I(c)$}\\
s(v)_\bM \ \ & = \ \ & s(v) \mbox{ as given in the model $\bM$} \\
s(x|_\varphi y)_\bM  \ \ & = \ \ & s(x)_\bM  \,\,\,\,\, \mbox{ if $s\models_\bM\varphi$, and }\\
s(x|_\varphi y)_\bM  \ \ & = \ \ & s(y)_\bM  \,\,\,\,\, \mbox{ if $s\not\models_\bM\varphi$ }\\
s(F(x_1,\ldots, x_n))_\bM  \ \ & = \ \ & (I(F)) (s(x_1)_\bM, \ldots, s(x_n)_\bM)
\\
s(x_A^\varphi)_\bM \ \ & = \ \ &  \mbox{ the unique } d\in D \mbox{ s.t. } w(x)_\bM=d \mbox{ for all } w\simA s \mbox{ with }
w\models_\bM\varphi,\\
{}\ \ & {} \ \ & \mbox{ if such a unique $d$ exists, and}\\
s(x_A^\varphi)_\bM \ \ &  = \ \ & \uparrow_\bD   \,\,\,\,\,  \mbox { otherwise.}
\\
s(e(x))_\bM \ \ &  = \ \ & s(x)_{e(\bM)} \,\,\,\,\, \mbox{ if $s\models_\bM pre_e$, and}\\
s(e(x))_\bM \ \ &  = \ \ & \uparrow_\bD \,\,\,\,\,  \mbox { otherwise.}
\\
\\
\bM=(S, \sim, \bullet(\bullet)) \ \ &  \mapsto \ \ & e(\bM)=(S^e, \sim^e, \bullet(\bullet)_{e(\bM)}), \,\,\, \,\,\ \mbox{  where }
\\
S^e \ \ & = \ \ & \| pre_e\|_\bM \,\,\, = \{s\in S: s\models_\bM pre_e\}
\\
s\sim_a^e w  \ \ & \mbox{ iff } \ \ & s\sim_{e(a)} w \,\, \, \,\,\ \mbox{ (for $s,w\in S^e$)}
\\
s(v)_{e(\bM)} \ \ & = \ \ & s(post_e(v))_\bM \,\,\,\,\, \mbox{ (for $s\in S^e$)}.
\end{array}
\]

\begin{pr}\label{Preservation}
Let $s,w$ be states s.t. $s\simA w$ for some group $A\subseteq\Agents$, let $\varphi$ be some formula s.t. $\mathcal{A}(\varphi)\subseteq A$, and let $x\in Var$ be a term s.t. $\mathcal{A}(x)\subseteq A$. Then we have:
\begin{enumerate}
\item $s\models \varphi$ iff $w\models \varphi$;
\item $s(x)=w(x)$.
\end{enumerate}
\end{pr}

\vspace{-4mm}

\section{An Example: The Numbers' Game Puzzle}


There are two brothers, Alex and Bob, their sister Daniela, as well as an evil intruder Eve.  Each of Alex, Bob and Daniela is given a sheet of paper on which there is a secret natural number $n_a, n_b, n_d \in \{0, 1, 2, \ldots\}$ (Alex receives $n_a$, etc), so that each of them can see only his/her number. It is common knowledge that Daniela's number is the absolute difference of the other two numbers ($n_d= |n_a-n_b|$).


The following scenario ensues (in which all announcements are public and it is common knowledge that nobody lies):


Daniela shows to Alex her sheet of paper (containing the number $n_d$). It is common knowledge that this is happening, but nobody else can see the number.

Then Daniela asks Alex ``Do you know Bob's number?",  and Alex answers ``I don't know".

Next, Daniela also shows to Bob her sheet of paper (with her number $n_d$). Again, it is common knowledge that this is happening, but nobody else can see the number.

Then Daniela asks Bob ``Do you know Alex's number?" and Bob answers ``I don't know".

Finally, Alex says ``\emph{Now} I know Bob's number".

\medskip

\par\noindent\textbf{Exercise:}
\begin{enumerate}
\item Show that, at the end of the above scenario, it is common knowledge among all agents (including Eve) that $0< n_d\leq n_a < 2 \cdot n_d \leq n_b$.

\item Show that, at the end of the above scenario, all three numbers are common knowledge among Alice, Bob and Daniela, but it is also common knowledge that Eve doesn't know any of the numbers.
 \item Suppose that, after the above scenario, Daniela announces that her number is $n_d=1$. Show that, after that, it is common knowledge (among \emph{all} agents) that the numbers are $n_a=1$, $n_b=2$, $n_c=1$.
\end{enumerate}

\medskip\par\noindent
{\bf Modelling the Numbers' Game Example.}
We take $D=\mathbb{N}$ (the set of natural numbers) and work with four agents ($a,b,d,e$) and three basic local variables $n_a$, $n_b$, $n_d$, each in possession of the corresponding agent. We use addition $x +y $ as our only function symbol, and equality $x=y$, inequality $x\leq y$ and strict inequality $x<y$ as our predicate symbols (with their standard interpretations on $\mathbb{N}$. The set of states in the initial model $\bS^0$ is
$S^0:=\{s=(s_a,s_b, s_d)\in \mathbb{N}\times \mathbb{N} \times \mathbb{N}: s_a=s_b+s_d \mbox{ or } s_b=s_a+s_d\}$.
The epistemic equivalence relations $\sim_i \subseteq S^0\times S^0$ are given by putting $s\sim_i w$ iff $s_i =w_i,$ for each $i\in \{a, b,d\}$, while $\sim_e= S^0\times S^0$ is the universal relation on $S^0$. The value map is the obvious one:
$s(n_i):=s_i$, for $i\in \{a, b, d\}$.

Our scenario starts with the semi-public data-sharing event $!(a:d)$ (in which Daniela shows her sheet to Alex): this changes our model to an updated model $\bS^1$, obtained by resetting Alex's relation to the intersection $\sim_a' :=\sim_a \cap \sim_d$ of his old relation with Daniela's (and leaving everything else the same). The next step is the announcement $!(\neg K_a n_b)$ (in which Alex publicly announces that he doesn't know Bob's number): this deletes all states having $s_d=0$ or $s_a<s_d$ (since these states satisfy $K_a n_b$ in $\bS^1$), leaving everything else the same. So in the resulting model $\bS^2$, we are left only states $s=(s_a, s_b, s_d)$ with $0< s_d \leq s_a$ (and $|s_a - s_b|=s_d$). In the next round, we have another semi-public data-sharing event $!(b:d)$, which resets Bob' relation to the intersection $\sim_b' := \sim_b\cap \sim_d$, and leaving everything else the same in the resulting model $\bS^3$. After this, the public announcement $!(\neg K_b n_a)$ eliminates all states having $0 \leq s_b< 2 \cdot s_d$ (since these states satisfy $K_b n_a$ in $\bS^3$). So in the updated model $\bS^4$ we are left only with states with $0< s_d \leq s_a$ and $2\cdot s_d \leq s_b$. Finally, the public announcement $!(K_a n_b)$ eliminates all states with $2\cdot s_d \leq s_a$ (since those satisfy $\neg K_a n_b$ in $\bS^4$). Hence, the final model $\bS^5$ consists only of states satisfying $0< s_d \leq s_a< 2\cdot s_d \leq s_b$ (and $|s_a-s_b|=s_d$).  It is obvious that the formula $C_{\{a, b, d, e\}} (0< n_d\leq n_a < 2 \cdot n_d \leq n_b)$ is true in all the states of $\bS^5$, which takes care of item 1 in our Exercise.

For item 2, it is easy to see that the formula $C_{\{a,b,e\}} \overline{n}$ holds in all states of $\bS^5$, where $\overline{n}=(n_a, n_b, n_e)$ is the tuple consisting of our three local variables.
Finally, to deal with item 3, we perform on $\bS^5$ another public announcement $! (n_d=1)$, which deletes all states except for $s=(1,2,1)$.

\vspace{-4mm}

\section{Axiomatizations and Main Results}

\vspace{-2mm}

The proof system $\mathbf{DLKV}$ for the full dynamic logic $DLKV$ is in Table \ref{tb0}, where we used the abbreviations
$K_A^\theta\varphi:= K_A(\theta\to \varphi)$, $K_a\varphi:=K_{\{a\}}\varphi$,
$K_A^\varphi x := K_A^\varphi (x=x_A^\varphi)$, and $K_A x := K_A^\top x$. The proof system $\mathbf{LKV}$ is obtained by taking only the axioms and rules in groups (I)-(V) (restricted to the static fragment $LKV$).

\smallskip

\begin{Th}\label{CompLKV}
The proof system $\mathbf{LKV}$
is sound and complete for the static fragment $LKV$, and its logic $LKV$ is decidable.
\end{Th}

Using this result, we can establish completeness and decidability for the full dynamic logic $DLKV$:

\begin{Th}\label{CompDLKV}
The proof system $\mathbf{DLKV}$ in Table \ref{tb0} is sound and complete for the dynamic logic $DLKV$. Moreover, the logic $DLKV$ is provably co-expressive with the static logic $LKV$, and thus it is decidable.
\end{Th}

\smallskip

Soundness of $\mathbf{LKV}$ and $\mathbf{DLKV}$ are easy verifications, so we skip the details. The \emph{proofs of completenes and decidability} are intricate, and they form the content of the next two sections.
\begin{table}[h!]
\begin{tabular}{ll}
\hline
\textbf{(I)} & \textbf{Axioms and rules of classical propositional logic}
\\
 \textbf{(II)} & \textbf{Axioms for equality}:
\end{tabular}
\\
\begin{tabular}{ll}
(Reflexivity) & $x=x$
\ \\
(Indiscernability) & $\ux=\uy \to (P\ux \uz \leftrightarrow P\uy\uz)$
\ \\
(Functionality) & $\ux=\uy \to F(\ux)=F(\uy)$
\ \\
(Definition by Cases) & $(\varphi \, \to\, x|_\varphi y= x) \wedge (\neg\varphi \, \to \, x|_\varphi y =y)$
\end{tabular}
\\
\begin{tabular}{ll}
\textbf{(III)} & \textbf{Axioms and rules for (distributed) knowledge}:
\end{tabular}
\\
\begin{tabular}{ll}
 (Necessitation) &  From $\varphi$, infer $K_A \varphi$
 \ \\
(Distribution) & $K_A (\varphi\to \psi)\, \to \, (K_A \varphi\to K_A\psi)$
\ \\
(Veracity) & $K_A\varphi \, \to \, \varphi$
\ \\
(Positive Introspection) & $K_A\varphi \, \to \, K_A K_A\varphi$
\ \\
(Negative Introspection) & $\neg K_A\varphi \, \to \, K_A\neg K_A\varphi$
\ \\
(Group-Monotonicity) & $K_A \, \varphi \to \, K_B \varphi$, whenever $A\subseteq B$
\end{tabular}
\\
\begin{tabular}{ll}
\textbf{(IV)} & \textbf{Axioms and rules for conditional common distributed knowledge}:
\end{tabular}
\\
\begin{tabular}{ll}
 (C-Necessitation) &  From $\varphi$, infer $C_{\mathfrak A}^\theta \varphi$
 \ \\
(C-Distribution) & $C_{\mathfrak A}^\theta (\varphi\to \psi)\, \to \, (C_{\mathfrak A}^\theta \varphi\to C_{\mathfrak A}^\theta\psi)$
\ \\
(Fixed Point) & $C_{\mathfrak A}^\theta\varphi \,\, \to \,\, \left(\varphi \wedge \bigwedge_{A\in {\mathfrak A}} K_A^\theta C_{\mathfrak A}^\theta\varphi\right)$
\ \\
(Induction) & $C_{\mathfrak A}^\theta (\varphi \to  \bigwedge_{A\in {\mathfrak A}} K_A^\theta\varphi) \,\, \to \,\, (\varphi\to C_{\mathfrak A}^\theta\varphi)$
\end{tabular}
\\
\begin{tabular}{ll}
\textbf{(V)} & \textbf{Axioms for conditional knowledge of values}
\end{tabular}
\\
\begin{tabular}{ll}
(Non-vacuous Knowledge of Values) & $x_A^\varphi\!\!\downarrow \, \, \to \, \, \left(\, K_A^\varphi x \wedge \langle K_A\rangle \varphi \, \right)$
\ \\
(Knowledge of Constants \& Local Variables) & $K_a c \wedge K_a (v_a)$
\ \\
(Knowledge of Hypothetical Values) & $K_A x_A^\varphi$
\ \\
(Knowledge of Predicates) & $K_A\overline{x} \, \to \, (P \overline{x}  \to K_A P\overline{x})$
\ \\
(Knowledge of Functions) & $K_A^\varphi \overline{x} \, \to \, K_A^\varphi F(\overline{x})$
\ \\
(Known Equality) & $K_A^\theta (x=y) \, \to\, (K_A^\theta x \to K_A^\theta y)$
\ \\
(Anti-Monotonicity) & $K_A(\varphi \to \theta) \, \to \, (K_A^\theta x \to K_A^\varphi x)$
\ \\
\hline
\end{tabular}
\\
\begin{tabular}{ll}
\textbf{(VI)} & \textbf{Dynamic reduction axioms for propositional formulas}
\end{tabular}
\\
\begin{tabular}{ll}
($[e]$-Necessitation) &  From $\varphi$, infer $[e]\varphi$
 \ \\
($[e]$-Distribution) & $[e] (\varphi\to \psi)\, \to \, ([e]\varphi\to [e]\psi)$
\ \\
(Atomic Change) & $[e] Px_1\ldots x_n \, \leftrightarrow \, \left( pre_e \to P e(x_1)\ldots e(x_n) \right)$
\ \\
(Partial Functionality)  & $[e] \neg\varphi \, \leftrightarrow \, \left( pre_e \to \neg [e] \varphi \right)$
\ \\
(Knowledge Update) & $[e] K_A\varphi \, \leftrightarrow \, \left( pre_e \to K_{e(A)} [e] \varphi \right)$
\ \\
($C_{\mathfrak{A}}$-Update) & $[e] C_{\mathfrak{A}}^\theta \varphi \, \leftrightarrow \, \left( pre_e \to C_{e[\mathfrak{A}]}^{\langle e\rangle \theta} [e] \varphi \right)$
\end{tabular}
\\
\begin{tabular}{ll}
\textbf{(VII)} & \textbf{Dynamic reduction axioms for terms}
\end{tabular}
\\
\begin{tabular}{ll}
(Survival of Value) & $e(x)\!\!\downarrow \, \, \to \, \, pre_e$ \ \\
(Preservation of Constants) & $pre_e \, \, \to \, \, e(c) \, = \, c$
\ \\
(Change of Basic Values) & $pre_e  \, \, \to \, \, e(v)  \, = \, post_e(v)$
\ \\
(Change of Disjunctive Terms) & $pre_e \, \, \to \, \, e(x|_\varphi y)  \, = \, e(x)|_{\langle e\rangle \varphi} e(y)$
\ \\
(Change of Functional Terms) & $pre_e \, \, \to \, \, e(F(x_1, \ldots, x_n)) \, = \, F(e(x_1), \ldots, e(x_n))$
\ \\
(Change of Hypothetical Values) & $pre_e  \, \, \to \, \, e(x_A^\varphi)  \, = \, e(x)_{e(A)}^{\langle e \rangle \varphi}$
\ \\
\hline
\end{tabular}
\caption{The Proof System $\mathbf{DLKV}$. The Proof System $\mathbf{LKV}$ consists only of the axioms and rules in groups (I)-(V) (restricted to the static fragment $LKV$).
}\label{tb0}
\end{table}

\newpage

To illustrate the power of our axiomatic system, we note
a number of interesting theorems.

\begin{pr}\label{theorems}
The following are provable in the system $\mathbf{LKV}$:
\begin{itemize}
\item (Strong Introspection) $\varphi \to K_A\varphi$, whenever $\mathcal{A}(\varphi)\subseteq A$
\item (Term Introspection)  $K_A x$, whenever $\mathcal{A}(x)\subseteq A$
\item  (Explicit Value Introspection) $\langle K_A\rangle \varphi \to x_A^\varphi=x$, whenever $\mathcal{A}(x)\subseteq A$
\item (Conditional Knowledge) $K_A^\varphi \psi  \leftrightarrow K_A (\varphi\to \psi)$
\item (Non-Vacuous Knowledge) $\langle K_A\rangle \varphi \to \left( K_A^\varphi x  \leftrightarrow  x_A^\varphi\!\!\downarrow \right)$
\item (True Conditional Value) $(x_A^\varphi\!\!\downarrow \wedge \varphi)\to x_A^\varphi=x)$
\item $K_A(\varphi \to x=y) \to x_A^\varphi =y_A^\varphi$
\item $\left( K_A (\varphi \to x=y) \wedge \varphi\wedge x\!\!\downarrow \right)\, \to \, x_A^\varphi\!\!\downarrow$,  whenever $\mathcal{A}(y)\subseteq A$
\end{itemize}
\end{pr}

\vspace{-4mm}

\section{Completeness and Decidability for the Static Logic $LKV$}\label{A}
\vspace{-2mm}
Throughout this section, \emph{we fix a formula $\varphi_0\in Fml$}. We'll prove Theorem \ref{CompLKV} by
the method of \emph{quasi-models}.
To obtain an appropriate analogue of Fischer-Ladner closure, we need an auxiliary notion:

\smallskip
\par\noindent\textbf{Restricted Vocabulary and Restricted Terms} For any finite set of formulas $\Sigma\subseteq Fml$, the \emph{$\Sigma$-restricted set of terms} $Var_\Sigma$ and the \emph{$\Sigma$-restricted vocabulary}
$\mathcal{V}_\Sigma :=  =(V_\Sigma, \Agents_\Sigma, ag_\Sigma, Pred_\Sigma, Funct_\Sigma, ar_\Sigma)$ are formed as follows: $Var_\Sigma$ is \emph{the set of all subterms (of any term) occurring in formulas of $\Sigma$};
$V_\Sigma:= V\cap Var_\Sigma$ is \emph{the set of all basic variables occurring (as subterms of any term) in formulas of $\Sigma$}; $\Agents_\Sigma$ is the set of agents occurring (inside terms or modalities in formulas) in $\Sigma$ or as locations $ag(v)$ of basic variables $v\in V_\Sigma$; $Pre_\Sigma$ is the set of predicates (occurring in formulas) in $\Sigma$; $Funct_\Sigma$ is the set of function symbols occurring in (terms of) $Var_\Sigma$; $ag_\Sigma$ is the restriction of $ag$ to $V_\Sigma$; and $ar_\Sigma$ is the restriction of $ar$ to $Pred_\Sigma\cup Funct_\Sigma$.

Clearly, if $\Sigma$ is finite, then $Var_\Sigma$ and all the sets in $\mathcal{V}_\Sigma$ are finite. But note that the $\Sigma$-restricted set of terms $Var_\Sigma$ is only \emph{a finite subset} of the (typically infinite) set $Var(\mathcal{V}_\Sigma)$ of terms of the language $LKV(\mathcal{V}_\Sigma)$.

\smallskip
\par\noindent\textbf{Closure} Given our fixed formula $\varphi_0$, \emph{the closure of $\varphi_0$} is a finite set $\Sigma=\Sigma(\varphi_0)\subseteq Fml$, defined as the smallest set of formulas satisfying the following properties: $\varphi_0\in \Sigma$; $\Sigma$ is closed under subformulas and single negations $\sim\varphi$; if $(K_A\varphi)\in \Sigma$ and $B\subseteq\Agents_\Sigma$, then $(K_B\varphi)\in \Sigma$; if $K_A(\varphi\to\psi)\in \Sigma$, then $(K_A\varphi)\in \Sigma$ (and so also $(K_B\varphi)\in \Sigma$;
if $(x|_\varphi y)\in Var_\Sigma$, then $\varphi\in \Sigma$;
if $(x_A^\varphi), y, z\in Var_\Sigma$, then $K_A(\varphi \to y=z) \in\Sigma$ (and so also
$(K_A^\varphi x), \langle K_A\rangle (\varphi\wedge y\neq z), \langle K_A\rangle \varphi, \varphi\in\Sigma$);
if $\ux=(x_1, \ldots, x_n)$ has all $x_1, \ldots, x_n\in Var_\Sigma$ and $P\in Pred_\Sigma$, then $P\ux, (\ux \!\!\downarrow)\in \Sigma$ (-and so in particular, if $x,y\in Var_\Sigma$
then $(x=y), x\!\!\downarrow, x\!\!\uparrow\in \Sigma$ and $\uparrow\in Var_\Sigma$).

\smallskip
\par\noindent\textbf{Types} Let $\Sigma$ be the closure of $\varphi_0$. A \emph{$\Sigma$-type} is a subset of $\Sigma$ with the following properties:
\begin{enumerate}
\item\label{neg} for every $\varphi\in\Sigma$: $(\sim\varphi)\in \Delta$ iff $\varphi\not\in\Delta$;
\item\label{conj} for every $(\varphi\wedge \psi)\in \Sigma$: $(\varphi\wedge \psi)\in \Delta$ iff $\varphi\in\Delta$ and $\psi\in \Delta$;
\item\label{eq1} for every $x\in Var_\Sigma$: $(x=x)\in\Delta$;
\item\label{eq4} for every  $P\ux\uz\in \Sigma$: if $(\ux=\uy), P\ux\uz\in \Delta$ then $P\uy\uz\in \Delta$;\footnote{This, together with clause \ref{eq1}, has the consequence that \emph{$\Delta$-equality is an equivalence relation on variables in $Var_\Sigma$}: $(x=y)\in \Delta$ implies $(y=x)\in \Delta$; and $(x=x'), (x'=x'')\in \Delta$ imply $(x=x'')\in \Delta$.}
\item\label{eq5} for $F(\ux), F(\uy)\in Var_\Sigma$: if  $(\ux=\uy)\in \Delta$, then $(F(\ux)=F(\uy))\in \Delta$;
\item\label{cases} for every $(x|_\varphi y)\in\Sigma$: $\varphi\in \Delta$ implies $(x|_\varphi y)\in\Delta$; and $\varphi\not\in\Delta$ implies $(x|_\varphi y=y)\in \Delta$;
\item\label{prop-veracity} if $(K_A\varphi)\in\Delta$ then $\varphi\in\Delta$;
\item\label{vac-know}  if $(x_A^\varphi\!\!\downarrow)\in \Delta$, then $K_A^\varphi x, \langle K_A\rangle\varphi\in\Delta$;
\item\label{veracity} if $\varphi, (x_A^\varphi\!\!\downarrow)\in\Delta$, then $(x=x_A^\varphi)\in\Delta$;
\item\label{hypo-eq} if $K_A(\varphi\to x=y)\in \Delta$ then $(x_X^\varphi=y_A^\varphi)\in \Delta$;
\item\label{undefined} $(c_A^\varphi=c)\in \Delta$ or $(c_A^\varphi=\uparrow)\in \Delta$, for all $c\in C \cap Var_\Sigma$;
\item\label{diff} for every $x_A^\varphi\in Var_\Sigma$ and every $y\in Var_\Sigma$ s.t. $\mathcal{A}(y)\subseteq A$: if $(x_A^\varphi\!\!\uparrow), \varphi, x\!\!\downarrow \in \Delta$, then $\langle K_A \rangle (\varphi\wedge x\neq y)\in \Delta$;
\item\label{common} if $(C_{\gA}^\theta\varphi)\in \Delta$ and $A\in \gA$, then $\varphi, (K_A^\theta C_{\gA}^\theta\varphi)\in\Delta$.
    \end{enumerate}

\par\noindent\textbf{Accessibility relations on types}
For types $\Delta, \Delta'$ and group $A\subseteq \Agents_\Sigma$, we put:
\vspace{-2mm}
$$\Delta \simA \Delta' \,\, \mbox{ iff } \,\, \varphi\in \Delta \Leftrightarrow \varphi\in\Delta' \, \mbox{ for all } \varphi\in\Sigma \mbox{ s.t. } \mathcal{A}(\varphi)\subseteq A.
\vspace{-2mm}$$
\begin{pr}\label{Access}
The relations $\simA$ are equivalence relations on types, satisfying Group Monotonicity:
\vspace{-2mm}
$$\Delta\simA \Delta' \mbox{ and } B\subseteq A \mbox{ imply } \Delta\simB \Delta'.\vspace{-2mm}$$
\end{pr} \vspace{-2mm}
\begin{proof} \vspace{-2mm}
This follows directly from the definition of relations $\simA$ on types.
\end{proof}

\vspace{-3mm}

\begin{pr}\label{Access2}
If $\Delta\simA \Delta'$ and $(K_A\varphi)\in \Delta$, then $\varphi\in \Delta'$.
\vspace{-2mm}
\end{pr}
\begin{proof}
\vspace{-1mm}
Since $\mathcal{A}(K_A\varphi)=A$ and $\Delta\simA \Delta'$, we use the definition of $\simA$ on types and the fact that
$(K_A\varphi)\in \Delta$ to infer that $(K_A\varphi)\in \Delta'$. This together with condition \ref{prop-veracity} on types, gives us the desired conclusion.
\end{proof}


\par\noindent\textbf{Hat notation}. For every type $\Delta$, we put
$\widehat{\Delta} \, \,\,\, :=\,\, \, \, \bigwedge\Delta$ for the \emph{conjunction of all formulas in $\Delta$}.

\begin{pr}\label{Access3}
Given types $\Delta$ and $\Lambda$, if $\widehat{\Delta}\wedge \langle K_A \rangle \widehat{\Lambda}$ is consistent, then $\Delta \simA \Lambda$.
\vspace{-1mm}
\end{pr}
\vspace{-1mm}
\begin{proof} \vspace{-2mm} Let $\Delta$ and $\Lambda$ be types as above, and suppose towards a contradiction that we have $\Delta \not\simA \Lambda$. Then there must exist $\varphi\in \Sigma$ s.t. $\mathcal{A}(\varphi)\subseteq A$ and $\varphi\in\Delta$, but $(\sim\!\varphi)\in\Lambda$. From this together with the assumption that $\widehat{\Delta}\wedge \langle K_A \rangle \widehat{\Lambda}$ is consistent, we infer that $\varphi\wedge \langle K_A \rangle \neg\varphi$ is consistent, and thus $\varphi\wedge \neg K_A\varphi$ is consistent.

But this contradicts the fact that $\vdash\, \varphi\to K_A\varphi$ is an $\mathbf{LDK}$-theorem for formulas $\varphi$ with $\mathcal{A}(\varphi)\subseteq A$.\end{proof}

\par\noindent\textbf{Quasi-Models}
A \emph{quasi-model for $\varphi_0$} is a set $S$ of types over $\Sigma=\Sigma(\varphi_0)$, with the following properties:
\begin{description}
\item[(*)] $\varphi_0\in \Delta_0$ for some type $\Delta_0\in S$;
\item[(**)] if $\langle K_A \rangle \varphi\in \Delta\in S$, then there is some $\Delta'\in S$ with $\Delta\simA \Delta'$ and $\varphi\in \Delta'$;
    \item[(***)] if  $\langle C_{\mathfrak A}^\theta \rangle \varphi\in \Delta\in S$, then there is some finite chain
    $\Delta=\Delta_0\sim_{A^1} \Delta_1 \sim_{A^2}\cdots \sim_{A^n} \Delta_n$,
    with $0\leq n \leq |S|-1$, $A^1, \ldots, A^n\in \gA$, $\theta\in \Delta_1, \ldots, \Delta_n\in S$ and $\varphi\in \Delta_n$.
\end{description}

\par\noindent\textbf{The $\Sigma$-canonical quasi-model} A \emph{$\Sigma$-theory} is any maximally consistent subset of $\Sigma$.  We denote by
$S_\Sigma$ the \emph{set of all $\Sigma$-theories}. We will \emph{show that $S_\Sigma$ is a (finite) ``canonical" quasi-model} for $\Sigma$.

\begin{lemma}\label{Types}
Every $\Sigma$-theory is a $\Sigma$-type.
\vspace{-1mm}
\end{lemma}
\begin{proof}
\vspace{-2mm}
This is easy to check, using the axioms of $\mathbf{LKV}$ and Proposition \ref{theorems}.
\end{proof}

\vspace{-2mm}

\begin{lemma}\label{K-ExistenceLemma}
For every $\Delta\in S_\Sigma$, if $\langle K_A \rangle \varphi\in \Delta$ then there is some $\Delta'\in S_\Sigma$ with $\Delta\simA \Delta'$ and $\varphi\in \Delta'$.
\end{lemma}
\begin{proof}
\vspace{-2mm}
Put
$\Delta_A\, :=\, \{\theta: \theta\in\Delta \mbox{ s.t. } \mathcal{A}(\theta)\subseteq A\}
\cup \{\sim \theta: \theta\in (\Sigma -\Delta) \mbox{ s.t. } \mathcal{A}(\theta)\subseteq A\}$.

\emph{Claim}: $\Delta_A \cup \{\varphi\}$ is consistent wrt the system $\mathbf{LKV}$.
\\
\emph{Proof of Claim}: Suppose not. Then we have $\vdash \widehat{\Delta_A}\to \sim \varphi$. Applying Necessitation and Distribution, we obtain $\vdash K_A\widehat{D_A} \to K_A\sim \varphi$. On the other hand, by inspecting the structure of $\Delta_A$, it is easy to see that $\mathcal{A}(\Delta_A)\subseteq A$, so by Strong Introspection (Proposition \ref{theorems}9) we have $\vdash \widehat{\Delta_A} \to K_A \widehat{\Delta_A}$. Putting these together, we get $\vdash \widehat{\Delta_A}\to K_A\sim \varphi$. Since $\Delta_A\subseteq \Delta$ and $\Delta$ is closed under $\Sigma$-consequences, we have $(K_A\sim\varphi)\in \Delta$. But this contradicts the assumption that $\langle K_A \rangle \varphi\in \Delta$ (given the consistency of $\Delta$).
\\
Using our Claim and the standard Lindenbaum Lemma, we get that $\Delta_A \cup \{\varphi\}$ has a $\Sigma$-maximally consistent extension $\Delta'\in S_\Sigma'$. So we have $\varphi\in \Delta'$ and $\Delta_A\subseteq \Delta'$, which implies that $\Delta\sim_A \Delta'$.
\end{proof}

\begin{lemma}\label{C-ExistenceLemma}
For every $\Sigma$-theory $\Delta\in S_\Sigma$, if  $\langle C_\gA^\theta \rangle \varphi\in \Delta$, then there is some finite chain of $\Sigma$-theories
    $\Delta=\Delta_0\sim_{A^1} \Delta_1 \sim_{A^2}\cdots \sim_{A^n} \Delta_n$,
    with $0\leq n \leq |S_\Sigma|-1$, $A^1, \ldots, A^n\in \gA$, $\theta\in \Delta_1, \ldots, \Delta_n\in S_\Sigma$ and $\varphi\in \Delta_n$.
\end{lemma}
\begin{proof}
\vspace{-2mm}
Put
$\Delta_{\gA}^\Delta \, :=\, \{\Gamma\in S_\Sigma: \Delta=\Gamma_0\sim_{A^1}\Gamma_1 \cdots \sim_{A^n} \Gamma_n=\Gamma
\mbox{ s.t.} n\geq 0, \theta\in \Gamma_k\in S_\Sigma \mbox{ and } A^k\in\gA \mbox{ for all } k\geq 1\}$.
We need to prove that there exists $\Gamma \in \Delta_{\gA}^\Delta$ s.t. $\varphi\in \Gamma$.\footnote{This shows that there exists a finite $\gA$-chain of $\Sigma$-theories $\Delta=\Gamma_0, \Gamma_1, \ldots, \Gamma_n=\Gamma$, with $n\geq 0$, $\varphi\in\Gamma$ and with $\theta\in \Gamma_k$ for all $j\geq 1$. The fact that we can take $n\leq |S|-1$ follows from the fact that there are no more than $|S_\Sigma|$ distinct such theories, so if the chain is longer we can just shorten it by cutting the repeating theories (together with the subchains connecting them).}
Suppose not. Then $(\sim\varphi)\in \Gamma$ for all $\Gamma \in \Delta_{\gA}^\Delta$. Take the formula
$\eta :=\bigvee\{ \widehat{\Gamma}: \Gamma \in \Delta_{\gA}^\Delta\}$. So we have $\vdash \eta \to \sim\varphi$, thus also $\vdash C_\gA^\theta \eta \to C_\gA^\theta \sim\varphi$.
\smallskip

\emph{Claim}: $\vdash \eta \to C_{\gA}^\theta \eta$.
\\
\emph{Proof of Claim}: Using the Induction axiom and $C_\gA^\theta$-Necessitation, we see that is enough to show that
$\vdash \eta \to K_A^\theta \eta$ for all $A\in \gA$. Suppose not: then we have that $\eta\wedge \langle K_A^\theta\rangle \neg \eta$ is consistent, for some $A\in \gA$. Unfolding the structure of $\eta$ and using the (easily verifiable) fact that $\vdash \bigvee \{\widehat{\Gamma}: \Gamma\in S_\Sigma\}$ is a theorem, we see that there must exist $\Sigma$-theories $\Gamma \in \Delta_\gA^\theta$ and $\Gamma'\in (S_\Sigma - \Delta_\gA^\theta)$ s.t. $\widehat{\Gamma} \wedge \langle K_A^\theta\rangle \widehat{\Gamma'}$ is consistent, i.e.
$\widehat{\Gamma} \wedge \langle K_A\langle (\theta\wedge \widehat{\Gamma'})$ is consistent. This means that on the one hand $\theta\wedge \widehat{\Gamma'}$ is consistent, hence $\theta\in \Gamma'$ (since $\theta\in \Sigma$ and $\Gamma'$ is a maximally consistent subset of $\Sigma$); and on the other hand, $\widehat{\Gamma} \wedge \langle K_A\rangle \widehat{\Gamma'}$ is consistent, hence by Proposition \ref{Access3} (and the fact that $\Gamma$ and $\Gamma'$ are types by Lemma \ref{Types}) we have $\Gamma\sim_A \Gamma'$. But given the way we defined the set $\Delta_\gA^\theta$, we see that from $\Gamma\in \Delta_\gA^\theta$, $A\in \gA$, $\Gamma\sim_A \Gamma'$ and $\theta\in \Gamma'$ we can infer that $\Gamma'\in \Delta_\gA^\theta$, which contradicts the fact that $\Gamma'\in (S_\Sigma - \Delta_\gA^\theta)$.

Using now our Claim and the fact that $\vdash C_\gA^\theta \eta \to C_\gA^\theta \sim\varphi$, we obtain that $\vdash \eta \to
C_\gA^\theta \sim\varphi$. But we also have $\vdash \widehat{\Delta}\to \eta$ (by the structure of $\eta$ and the fact that $\Delta\in \Delta_\gA^\theta$ by definition), hence we infer $\vdash \widehat{\Delta}\to  C_\gA^\theta \sim\varphi$. But this contradicts the assumption that
$\langle C_\gA^\theta \rangle \varphi\in \Delta$ (given the consistency of $\Delta$).
\end{proof}
\vspace{-2mm}
\begin{pr}\label{QuasiModel}
If $\varphi_0\in \Sigma$ is consistent, then there exists a quasi-model for $\varphi_0$.
\end{pr}
\begin{proof}
\vspace{-2mm}
Take the set $S_\Sigma$ of all $\Sigma$-theories. By the Lindenbaum Lemma, there exists a maximally consistent subset $\Delta_0\in S_\Sigma$, such that $\varphi_0\in \Delta_0$. By Lemmas \ref{Types}, \ref{K-ExistenceLemma} and \ref{C-ExistenceLemma}, $S_\Sigma$ is a quasi-model for $\varphi_0$.
\end{proof}
\vspace{-2mm}
Putting this together with the \emph{soundness} of the proof system $\mathbf{LKV}$, we obtain the following:

\smallskip
\par\noindent\textbf{Corollary}. \emph{If $\varphi_0$ is satisfiable then there exists a quasi-model for $\varphi_0$.}

\smallskip

The hard part is to prove the \emph{converse} of this:

\begin{pr}\label{Satisfiability}
If $S$ is a quasi-model for $\varphi_0$, then $\varphi_0$ is satisfiable.
\end{pr}

Note that this would immediately give us the completeness and decidability of $\mathbf{LKV}$:

\smallskip\par\noindent\emph{\textbf{Proof of Theorem \ref{CompLKV}}}:
By Proposition \ref{Satisfiability}, Proposition \ref{QuasiModel} and the soundness of the system $\mathbf{LKV}$, the following three notions are equivalent: (a) $\varphi_0$ is consistent; (b) there exists a quasi-model for $\varphi_0$; (c) $\varphi_0$ is satisfiable. \emph{Completeness} immediately follows. For decidability: to decide whether $\varphi_0$ is satisfiable, take the closure $\Sigma=\Sigma(\varphi_0)$ of $\varphi_0$, and check whether $\varphi$ belongs to any type in a quasi-model over $\Sigma$.\footnote{Note that there are only finitely many types over $\Sigma$, so only finitely many quasi-models, and in fact we can calculate an upper bound on their number, that is a computable function of the complexity of $\varphi_0$. Equally important, checking that a finite subset of $\Sigma$ is a type, and that a finite set of types is a quasi-model, are both computable tasks.}

\medskip

The rest of this section is dedicated to the proof of Proposition \ref{Satisfiability}.

\bigskip

\par\noindent\textbf{Unravelling: the tree of histories} Let us fix a quasi-model $S$, a formula $\varphi_0$ and a type $\Delta_0\in S$ with $\varphi_0\in \Delta_0$. We will construct a model for $\varphi$, essentially based on an unravelling of $S$ around $\Delta_0$.
A \emph{history} is a finite sequence $h=(\Delta_0, A^1, \Delta_1, \ldots, A^n, \Delta_n)$ of any length $n\geq 0$, where $\Delta_1, \ldots,\Delta_n\in S$ are types and
$A^1, \ldots, A^n\subseteq\Agents_\Sigma$ are groups, such that we have
$\Delta_{i-1}\sim_{A^i} \Delta_i$ for all $i=1,n$.
We denote by $last(h):=\Delta_n$ the last state in history $h$, and by $\to_{A}$ the natural \emph{one-step relation} on histories, given by putting: $h\to_A h'$ iff $h'=(h, A, \Delta')$ (with $last(h) \simA \Delta'=last(h')$).
The one-step relations structure $H$ into a \textit{tree rooted at $\Delta_0$} (with the immediate successor relation given by $h\to h'$ iff $h\to_A h'$ for some group $A$). In particular, we have \textit{the tree property} : \emph{every two nodes $h, h'$ of the tree are connected by a unique non-redundant path} $h=h_0 \ot_{A^1} h_1 \ot_{A^2} \ldots \ot_{A^i} h_{i}\to_{A^{i+1}} \ldots \to_{A^n} h_n=h'$ (-in which neighboring nodes on the path are immediate successors of one another, in one order or another, and no nodes are repeated).

\smallskip
\par\noindent\textbf{Epistemic relations on histories}
To make this tree into a model for our restricted vocabulary $\mathcal{V}_\Sigma$, we define our \emph{single-agent indistinguishability relations} $\sima\subseteq H\times H$ on histories, by putting
\vspace{-2mm}
$$\sima \,\, :=\,\, \left( \bigcup_{A\ni a} \to_A \cup \bigcup_{A\ni a} \ot_A \right)^*,
\vspace{-2mm}$$
where $\ot_A$ is the converse of $\to_A$, the unions range over groups $A\subseteq \mathcal{A}_\Sigma$ s.t. $a\in A$, and $R^*$ is the reflexive-transitive closure of $R$.
Since our goal is to build a standard model, the \emph{group indistinguishability relations} $\simA\subseteq H\times H$ are taken to be simply the \emph{intersections} of all the individual relations:
\vspace{-2mm}
$$\simA \,\, :=\,\, \bigcap_{a\in A} \sima
\vspace{-2mm}$$
It is useful to give a more concrete characterizations of the knowledge relations $\sima$ and $\simA$ on histories.

\begin{lemma}\label{b-path}
The following are \emph{equivalent}, for all $a\in \mathcal{A}_\Sigma$ and histories $h,h'\in H$:
\begin{enumerate}
\item $h\sima h'$;
\item $a\in A^i$, for all groups $A^i$ that appear as transition labels on the non-redundant path from $h$ to $h'$.
\end{enumerate}
\vspace{-2mm}
\end{lemma}
\vspace{-1mm}
\begin{proof}
\vspace{-2mm}
This is obvious, by the definition of $\sima$ on $H$, and the uniqueness of the non-redundant path.
\end{proof}

\begin{lemma}\label{path}
The following are \emph{equivalent}, for $A\subseteq \mathcal{A}_\Sigma$ and histories $h,h'\in H$:
\begin{enumerate}
\item $h\simA h'$;
\item $A\subseteq A^i$, for all groups $A^i$ that appear as transition labels on the non-redundant path from $h$ to $h'$.
\end{enumerate}
\end{lemma}
\vspace{-1mm}
\begin{proof}
\vspace{-2mm}
This follows from Lemma \ref{b-path}, by the uniqueness of the non-redundant path from $h$ to $h'$.\end{proof}

\vspace{-2mm}

Next, we show that the map $last$ preserves $\simA$-relations:

\begin{lemma}\label{equiv}
If $h\simA h'$, then $last(h)\simA last(h')$.
\end{lemma}
\vspace{-1mm}
\begin{proof}
\vspace{-2mm}
The proof is by \emph{induction on the length $N$ of the non-redundant path} from $h$ to $h'$:

\emph{Base case}: $h=h'$. The conclusion follows trivially (given that $\simA$ are equivalence relations).

\emph{Inductive case}: Suppose the non-redundant path from $h$ and $h'$ has length $N+1$, and let us look at the last transition on this path. Given Lemma \ref{path}, this transition can be either of the form
$h_N {\to}_{A^N} h_{N+1}=h'$, or of the form $h_N {\ot}_{A^N} h_{N+1}=h'$, with $A^N\supseteq A$. By definition of $\to_A$ on histories, we have either  $h'=(h_N, A^N, last(h'))$ or  $h_N=(h', A^N, last(h_N))$, with $last(h_N)\sim_{A^N} last (h')$ in both cases. By Monotonicity (Proposition \ref{Access}) and the fact that $A\subseteq A^N$, we obtain $last(h_N)\simA last(h')$. On the other hand, we also have $last(h)\simA last(h_N)$ (-since the non-redundant path from $h$ to $h_N$ has length $N$, so by the induction hypothesis the pair $(h,h_N)$ satisfies the conclusion of our Lemma, with $h'$ replaced by $h_N$). Putting these two together (and using the transitivity of $\simA$), we conclude that $last(h)\simA last(h')$, as desired.
\end{proof}

\par\noindent\textbf{The Value Domain: a Quotient Construction} The \emph{set of values} $D$ of our model will be given by a quotient of the Cartesian product $H\times Var_\Sigma$.
We \emph{will define an equivalence relation}
$\approx$ on pairs $(history, variable)$ in $H\times Var_\Sigma$ (telling us \emph{when two such pairs represent the same value}), then we'll define our canonical set of objects $D$ to be \emph{the quotient of $H\times Var_\Sigma$ with respect to $\approx$}.

For this, we first introduce another \emph{equivalence relation} $\sim$ on $H\times Var_\Sigma$, representing ``\emph{identity of objects at a given node}". This is given by putting:
\vspace{-2mm}
$$(h,x) \sim (h', x') \,\, \mbox{ iff } \,\, h=h' \mbox{ and } (x= x')\in last(h).
\vspace{-2mm}$$
\begin{pr}\label{Sim} $\sim$ is an equivalence relation.
\vspace{-1mm}
\end{pr}
\vspace{-1mm}
\begin{proof}
\vspace{-2mm}
This follows immediately from the fact that $\Delta$-equality is an equivalence relation for every type $\Delta$ (together with the fact that equality on histories is an equivalence relation).
\end{proof}

\par\noindent Second, we introduce a \emph{one-step relation $\to$ between pairs} in $H\times Var_\Sigma$, given by putting:
\vspace{-2mm}
$$(h,x) \to (h', x') \,\, \mbox{ iff } \,\, \exists y\in Var_\Sigma\, \exists A\supseteq \mathcal{A}(y)
\mbox{ s.t. } h\to_A h',
(x=y)\in last(h)  \mbox{ \& } (x'=y)\in last(h').\vspace{-2mm}$$
Finally, \emph{\emph{we define our main equivalence relation $\approx$ on pairs}} $(history, variable)$ in $H\times Var_\Sigma$, as the \emph{smallest equivalence relation on ${\mathcal D}$, that includes the relations $\sim$ and $\to$}.
It is useful to have a more concrete characterization of $\approx$. A first such characterization is immediate from the reflexivity and symmetry of $\sim$:
if we denote the converse of $\to\subseteq H\times Var_\Sigma$ by $\leftarrow$, then
it is easy to see that \emph{$\approx$ coincides with the transitive closure $(\sim\cup\to\cup \leftarrow)^+$ of the union of the relations $\sim$, $\to$ and $\leftarrow$}.
To obtain a more useful characterization, note that the relational compositions $\to; \sim$ and $\sim; \to$ coincide with $\to$. This, together with the existence of the non-redundant path from $h$ to $h'$, immediately give us the following:

\smallskip
\begin{lemma} \label{PathLemma} (``Path Lemma")
 Let $(h, x), (h', x')\in  H\times Var_\Sigma$, and let
$h=h_0 \ot h_1 \ot \ldots \ot h_i \to\ldots \to h_n=h'$
be the non-redundant path from $h$ to $h'$. Then the following are equivalent:
\begin{enumerate}
\item $(h,x)\approx(h',x')$;
\item there exist terms $x_0, x_1, \ldots, x_i, \ldots, x_n\in Var_\Sigma$ such that
we have:
$$(h,x)= (h_0, x_0)\leftarrow (h_1, x_1) \leftarrow (h_2, x_2)\leftarrow\cdots \leftarrow
(h_i,x_i)\to \cdots
\to(h_n, x_n)= (h', x')$$
\end{enumerate}
\end{lemma}

\medskip

\par\noindent\textbf{From Quasi-Model to Model} We are first defining our \emph{first-order data model} $\bD=(D, I)$ for the restricted vocabulary $\mathcal{V}_\Sigma$: as announced, the \emph{set of `values'} $D$ is the
quotient
\vspace{-2mm}
$$D\,\, :=\,\, (H\times Var_\Sigma)/\approx \, =\, \{[h,x]: (h,x)\in H\times Var_\Sigma\},\vspace{-2mm}$$
consisting of equivalence classes modulo $\approx$
\vspace{-2mm}
$$[h,x]:=\{(h', x')\in H\times Var_\Sigma: (h,x) \approx (h', x')\}  \, \, \mbox{ (for any given pair $(h,x)\in H\times Var_\Sigma$)}.\vspace{-2mm}$$
The \emph{interpretation function} $I$ will map every constant $c\in C\cap Var_\Sigma$ into
$I(c)= c_\bD  :=\, \, [(\Delta_0), c]=[h,c]$ (for all $h\in H$),
and map all $c\in C- Var_{\Sigma}$ into $I(c):= \uparrow_\bD$; it
will also map
 $n$-ary predicate symbols $P\in Pred_\Sigma$ into $n$-ary relations $I(P)\subseteq D^n$ given by
 \vspace{-2mm}
$$I(P) \, := \, \{([h,x_1], \ldots, [h,x_n])\in D^n: h\in H, \ux=(x_1, \ldots, x_n)\in Var_\Sigma^n \mbox{ s.t. } P \ux\in last(h)\};\vspace{-2mm}$$
and it will map $n$-ary functional symbols $F\in Funct_\Sigma$ into an $n$-ary functions $I(F): D^n\to D$ given by
\vspace{-2mm}
$$I(F) ([h,x_1], \ldots, [h,x_n]) \, :=\, [h, F(x_1, \ldots, x_n)] \, \,\, \, \mbox{ if $F(x_1, \ldots, x_n)\in Var_\Sigma$, and} \vspace{-2mm}$$
$$I(f) ([h,x_1], \ldots, [h,x_n]) \, :=\, \uparrow_\bD \, \,\, \, \mbox{ otherwise.}$$
Next, \emph{our epistemic data model} $\bM= (H, \sim, \bullet(\bullet))$ is given by taking: as set of states, the set $H$ of all histories; the indistinguishability relations $\sima\subseteq H\times H$ as defined above on histories\footnote{Note that this is meant to be a (standard) model, so the group indistinguishability relations $\simA$ are just
the intersections $\bigcap_{a\in A}\sima$, thus coinciding with the general relations $\simA$ introduced above.};
and the basic assignment function $\bullet(\bullet): H\times V_\Sigma \to D$ is given by putting $h(v_i) := [h,v_i]$, for all $h\in H$ and $v\in V_\Sigma$.
\medskip
\begin{lemma} \label{IntLemma} (``Interpretation Lemma")
The interpretation $I$ is well-defined, i.e. we have the following:
\begin{enumerate}
\item $(h,\ux)\approx (h', \uy)$ implies $(h, F(\ux))\approx (h', F(\uy))$;
\item  $(h,\ux)\approx (h', \uy)$ implies that: $(P\ux\,\uz)\in last (h)$ iff $(P\uy\, \uz)\in last (h')$;
\item $I(E)$ is really the identity relation on $D$:
$$([[h,x], [h', x'])\in I(E) \, \mbox{ iff } \, [h,x]=[h', x'] \, \mbox{ iff } \, (h,x)\approx (h',x')$$
\end{enumerate}
\vspace{-1mm}
\end{lemma}
\vspace{-1mm}
\begin{proof}
\vspace{-2mm}
The proofs are by induction on the length of the non-redundant path from $h$ to $h'$, using our conditions on types.\end{proof}

\begin{lemma} \label{Pres-Value}(``Preservation of Values")
If $h\simA h'$ are histories and $x\in Var_\Sigma$ is s.t. $\mathcal{A}(x)\subseteq A$, then $[h,x]=[h',x]$.
\vspace{-1mm}
\end{lemma}
\vspace{-1mm}
\begin{proof}
\vspace{-2mm}
Let
$h=h_0 \ot_{A^1} h_1 \ot_{A^2} \ldots \ot_{A^i} h_{i}\to_{A^{i+1}} \ldots \to_{A^n} h_n=h'$
be the non-redundant path from $h$ to $h'$. Since $h\simA h'$, we know that $A\subseteq A_k$ for all $k$ (by Lemma \ref{path}). Using this together with $\mathcal{A}(x)\subseteq A$ (and the definition of the relation $(h,x)\to (h',x')$ on history-variable pairs), we can see that we have
$(h,x)\leftarrow (h_1, x) \leftarrow\cdots
(h_i, x)\to \cdots
\to (h_n, x)= (h', x)$.
Thus, we have $(h,x)\approx (h', x)$, and hence $[h,x]= [h',x]$.\end{proof}

\begin{lemma} \label{Know-Value}(``Knowledge-of-Value Lemma")
If $(x_A^\varphi\!\!\downarrow)\in last(h)$, then there exists some history $h'\in H$, satisfying:
\begin{enumerate}
\item $h'\simA h$;
\item $\varphi\in last(h')$;
\item $[h'', x] = [h', x]= [h, x_A^\varphi]$ for all $h'' \simA h$ s.t. $\varphi\in last(h'')$.
\end{enumerate}
\vspace{-2mm}
\end{lemma}
\begin{proof}\vspace{-2mm}
From  $(x_A^\varphi\!\!\downarrow)\in last (h)$, we infer that $\langle K_A \rangle \varphi\in last(h)$ (by condition \ref{vac-know} in the definition of types), so by the Diamond Lemma there exists some $h'\simA h$ s.t. $\varphi\in last(h')$.

To prove the third claim of the lemma, let $h''\in H$ be s.t. $h''\simA h$ and $\varphi\in last(h'')$. From  $(x_A^\varphi\!\!\downarrow)\in last (h)$ and $h'', h'\simA h$, we obtain that  $(x_A^\varphi\!\!\downarrow)\in last (h')\cap last(h'')$ (by the definition of $\simA$ on histories and the fact that $\mathcal{A}(x_A^\varphi\!\!\downarrow)=A$). From this together with $\varphi \in last(h')\cap last(h'')$, we obtain that $(x=x_A^\varphi)\in last(h')\cap last(h'')$ (by condition \ref{veracity} on types), hence $(h',x)\approx (h', x_A^\varphi)$ and $(h'',x)\approx (h'', x_A^\varphi)$. On the other hand, by Lemma \ref{Pres-Value} and the fact that $\mathcal{A}(x_A^\varphi)=A$, we have $(h', x_A^\varphi)\approx (h,  x_A^\varphi)\approx (h', x_A^\varphi)$. Putting these together and using the transitivity and symmetry of the relation $\approx$, we obtain that $(h'', x)\approx (h', x)\approx (h, x_A^\varphi)$, i.e. $[h'', x]=[h', x]=]h, x_A^\varphi]$, as desired.\end{proof}

\begin{lemma} \label{Unknown-Value}(``Unknown Proper Value Lemma")
If $(x_A^\varphi\uparrow), \langle K_A\rangle (\varphi \wedge x\!\!\downarrow)\in last(h)$, then there exist histories $h'\simA h''\simA h$, with $\varphi\in last(h')\cap last(h'')$ and $[h', x]\neq [h,x]$.
\vspace{-1mm}
\end{lemma}
\begin{proof} \vspace{-1mm}
From $\langle K_A\rangle (\varphi \wedge x\!\!\downarrow)\in last(h)$, we infer by the Diamond Lemma that there exists $h'\simA h$ with $\varphi, x\!\!\downarrow\in last(h')$. If we put $\Delta:=last(h')$, we have $last(h)\simA \Delta$ and  $\varphi, x\!\!\downarrow\in \Delta$, and also $x_A^\varphi\!\!\uparrow\in \Delta$ (by the definition of $\simA$ on histories and the fact that $x_A^\varphi\!\!\uparrow\in last(h)\simA \Delta$, together with the fact that $\mathcal{A}(x_A^\varphi\!\!\uparrow)=A$). Thus, the assumptions in clause \ref{diff} on quasi-models are satisfied, except that we need to find some suitable $y$ to apply clause \ref{diff}.

Let us choose $y$ to be some randomly chosen variable $y\in Var_\Sigma$ s.t. $(x=y)\in \Delta$ and $\mathcal{y}\subseteq A$, if such a variable exists (Case 1); while otherwise (Case 2), we let $y:=\uparrow$. In both cases, we have $\mathcal{A}(y)\subseteq A$. Applying now clause \ref{diff} on types (to $\Delta$ and $y$), we must have that $\langle K_A\rangle (\varphi\wedge x\neq y)\in \Delta$. By condition (**) in the definition of pseudo-models, there exists some $\Delta'\in S$ with $\Delta'\simA \Delta$ and $\varphi, (x\neq y)\in \Delta'$.
Take now the history $h'':=(h',A, \Delta')=(h, A, \Delta, A, \Delta')$. It is clear that $h\simA h'\simA h''$ and that $\varphi\in \Delta\cap \Delta'= last(h')\cap last(h'')$. So, to prove our Lemma, we only need to check that $(h', x)\not\approx (h'',x)$.

For this, assume towards a contradiction that we have $(h', x)\approx (h'',x)$. Using the characterization of $\approx$ in the Path Lemma \ref{PathLemma} and the fact that the non-redundant path from $h'$ to $h''$ is $(h', \to^A, h'')$, we infer that there exists $y'\in Var_\Sigma$ with $\mathcal{A}(y')\subseteq A$, $(x=y')\in last(h')=\Delta$ and $(x=y')\in last (h'')=\Delta'$. The fact that we have  $(x=y')\in\Delta$ means that we are in Case 1 above, so that our previously chosen $y$ has the property that $(x=y)\in \Delta$. Using the fact
that $\Delta$-equality is an equivalence relations on variables in $Var_\Sigma$, we obtain that $(y'=y)\in \Delta$. Since $\mathcal{A}(y), \mathcal{A}(y')\subseteq A$, we have
$\mathcal{A}(y'=y)\subseteq A$, and so from  $(y'=y)\in \Delta$ and  $\Delta\simA \Delta'$ we infer that $(y'=y)\in \Delta'$ (by the definition  of $\simA$ on histories). Putting this together with the fact that  $(x=y')\in \Delta'$ and using again the $=$-transitivity condition on types, we conclude that $(x=y)\in \Delta'$. But this contradicts the fact that $(x\neq y)\in \Delta'$ (given condition \ref{neg} on types).\end{proof}

\begin{lemma} \label{Truth}(``Truth Lemma")
For every formula $\varphi\in \Sigma$ and every variable $x\in Var_\Sigma$, the following hold for all histories $h\in H$:
\begin{enumerate}
\item $h\models \varphi$ iff $\varphi\in last(h)$;
\item $h(x)=[h,x]$.
\end{enumerate}
\vspace{-2mm}
\end{lemma}
\begin{proof}
\vspace{-2mm}
We need a \emph{non-standard notion of ``complexity" order on both formulas $\varphi\in \Sigma$ and terms $x\in Var_\Sigma$}, according to which: (a) a formula is more complex than its subformulas and than the terms occurring in it; (b) a term is more complex than its subterms (variables); (c) the term $x_A^\varphi$ is more complex than the formulas $\langle K_A \rangle (\varphi\wedge x\!\! \downarrow), K_A(\varphi\to x=\uparrow)$.
It is easy to see that \emph{there indeed exists a well-founded partial order on $\Sigma$ satisfying these properties}.

\smallskip

We now prove the two claims \textit{simultaneously by induction on our non-standard ``complexity" order}.

\smallskip

\textbf{\emph{Proof of the Atomic Cases for Claims 1 and 2}}:
\smallskip

(i) \emph{Propositional Atoms}: $\varphi=P\ux$, with $\ux=(x_1, \ldots, x_n)$. We have the following equivalences:

\noindent $h \models P\ux$ iff $h(\ux)\in I(P)$ iff (by induction hypothesis for claim (2)) $([h,x_1], \ldots, [h, x_n])\in I(P)$ iff
(by definition of $I(P)$) $\exists h'\in H \mbox{ s.t. }  (\, (h,\ux)\approx (h', \ux) \& (P \ux)\in last(h') \, )$ iff
(by Lemma \ref{IntLemma}) $(P \ux)\in last(h)$.

(ii) \emph{Constants}: By definition, $w^h(c)=c_\bD=[h,c]$.

(iii) \emph{Atomic variables}: $x:=v_i\in V_\Sigma$ is a basic variable. This is trivial: by definition, $w^h (v_i)=[h,v_i]$.

\smallskip

\textbf{\emph{Proof of the Inductive Cases for Claims 1 and 2}}:
\smallskip

(iv) \emph{Boolean Cases}: $\neg\varphi, \varphi\wedge\psi$. These is trivial, using conditions \ref{neg} and \ref{conj} on types.

(v) \emph{$K_A$-modal Case}: $(K_A \varphi)\in \Sigma$. From \emph{left-to-right}: assume (towards a contradiction) that we have $h\models K_A\varphi$ but $(K_A\varphi)\not\in last(h)$. Then, by the closure condition \ref{neg} on types, we have $(\langle K_A\rangle\!\sim\!\varphi)\in last(h)$, and so by the property (**) of quasi-models, there exists some type $\Delta'\in S$ s.t. $last(h)\simA \Delta'$ and $(\sim\!\varphi)\in\Delta'$. Consider the history $h':=(h, A, \Delta')$: this is a well-defined history in $H$, satisfying $h \simA h'$ and $last(h')=\Delta'$. From $h\models K_A\varphi$ we infer that $h'\models \varphi$, and so by the induction hypothesis we have $\varphi\in last(h')=\Delta'$. But this contradicts the fact that $(\sim\varphi)\in\Delta'$, given condition \ref{neg} on types.
\\\
For the \emph{right-to-left} direction: we assume that $(K_A\varphi)\in last (h)$, and we have to prove that $h\models K_A\varphi$. For this, let $h'\in H$ be s.t. $h\simA h'$, and we need to show that $h'\models \varphi$. From $h\simA h'$ we obtain that $last(h)\simA last(h')$ (by Lemma \ref{equiv}), which together with $h\models K_A\varphi$ gives us $\varphi\in last(h')$ (by Proposition \ref{Access2}). Applying the induction hypothesis, we conclude that $h'\models \varphi$, as desired.

(vi) \emph{$C_{\gA}$-modal Case}: $(C_{\gA}^\theta \varphi)\in \Sigma$. From \emph{left-to-right}: assume (towards a contradiction) that we have $h\models C_{\gA}^\theta\varphi$ but $(C_{\gA}^\theta\varphi)\not\in last(h)$. Then, by the closure condition \ref{neg} on types, we have $(\langle C_{\gA}^\theta\rangle\!\sim\!\varphi)\in last(h)$, and so by the property (**) of quasi-models, there exists some finite chain of types
$\Delta_0\sim_{A^1}\Delta_1 \ldots \sim_{A^n}\Delta_n$, with $A^1, \ldots, A^n\in \gA$, $\Delta_0=last(h)$, $\theta\in\Delta_1, \ldots, \Delta_n\in S$ and $(\sim\!\varphi)\in \Delta_n$. Consider the histories $h_1=(h, A^1, \Delta_1)$, $h_2=(h_1, A^2, \Delta_2)=(h, A^1, \Delta_1, A^2, \Delta_2)$, $\ldots$, $h_n:=(h, A_1, \Delta_1, \ldots, A_n, \Delta_n)$: these are well-defined histories in $H$, satisfying $h=h_0 \sim_{A^1} h_1 \sim_{A^2} h_2\cdots \sim_{A^n} h_n$ and $\theta\in last(h_1), \ldots, last(h_n)$. By the induction hypothesis, we have $h_1, h_2, \ldots, h_n\models \theta$. Using this and the assumption that $h\models C_{\gA}^\theta\varphi$ (as well as the semantics of $C_{\gA}$), we infer that $h_n\models\varphi$. Applying again the induction hypothesis, we have $\varphi\in last(h_n)=\Delta_n$, but this contradicts the fact that $(\sim\!\varphi)\in\Delta$ (given the
consistency condition \ref{neg} on types).
\\\
For the \emph{right-to-left} direction: we assume that $(C_{\gA}^\theta\varphi)\in last (h)$, and we have to prove that $h\models C_{\gA}^\theta\varphi$. For this, let $h=h_0\sim_{A^1} h_1 \cdots h_{n-1}\sim_{A^n} h_n$ be any chain of histories (of any length $n\geq 0$) with $A^1, \ldots, A^n\gA$ and $h_1, \ldots, h_n\theta$, and we need to show that $h_n\models \varphi$. For this,
we first prove the auxiliary claim that $(C_{gA}^\theta\varphi)\in last(h_n)$, by induction on the length $n$ of the chain. For $n=0$, we have $h_0=h$, so we already have $(C_{\gA}^\theta\varphi)\in last (h)=last(h_0)$. For the inductive step for $n\geq 1$, we can assume by induction that
$(C_{\gA}^\theta\varphi)\in last(h_{n-1})$. Using condition \ref{common} on types, we infer that $(K_{A^n} C_{\gA}^\theta\varphi) \in last(h_{n-1})$. Using the fact that $h_{n-1}\sim_{A^n} h_n$ and Proposition \ref{Access2}, we obtain that $(C_{\gA}^\theta\varphi) \in last(h_{n})$, so we have established our auxiliary claim. From this, we infer that $\varphi\in last(h_n)$ (using again condition \ref{common} on types), and so the induction hypothesis for Claim 1, we conclude that $h_n\models \varphi$, as desired.

(vii) \emph{Terms defined by cases} $x|_\varphi y$. Given $(x|_\varphi y)\in\Sigma$, we have $\varphi\in\Sigma$, so we have that either $\varphi \in last(h)$, or else $(\sim\varphi)\in last(h)$. In the first case, we have $(x|_\varphi y=x)\in last(h)$ (by condition \ref{cases} on types), hence $[h, x|_\varphi]=[h,x]$; and on the other hand, $\varphi \in last(h)$ yields $h\models\varphi$ (by the induction hypothesis for Claim 1), so by definition we have $h(x|_\varphi y)=h(x)$, and by the induction hypothesis for Claim 2 we have $h(x)=[h,x]$, thus obtaining
$h(x|_\varphi y)= h(x)= [h,x]=h(x|_\varphi)$, as desired. The second case is similar: from $(\sim\varphi)\in last(h)$ we get
$(x|_\varphi y=y)\in last(h)$ (by condition \ref{cases} on types), hence $[h, x|_\varphi]=[h,y]$; and on the other hand, $(\sim\varphi) \in last(h)$ implies $\varphi\not\in last(h)$ (by the consistency of types), which by the induction hypothesis for Claim 1 yields $h\not\models\varphi$, so by definition we have $h(x|_\varphi y)=h(y)$, and by the induction hypothesis for Claim 2 we have $h(y)=[h,y]$, thus obtaining $h(x|_\varphi y)= h(y)= [h,y]=h(x|_\varphi)$, as desired.

(viii) \emph{Functional terms} $F(\ux)\in Var_{\Sigma}$ for some $\ux=(x_1,\ldots, x_n)$. We assume by the induction hypothesis that $h(\ux)=[h,\ux]$, and using the semantics of functional terms and the definition of $I(F)$ in our history model, we obtain
$h(F(\ux))= (I(F)) (h(\ux))= (I(F)) [h,\ux]= [h, F(\ux)]$, as desired.
%

(ix)  \emph{Conditional terms} $x_A^\varphi\in Var_{\Sigma}$. We distinguish three subcases.

\emph{Case 1}: $x_A^\varphi\!\!\uparrow, \langle K_A\rangle (\varphi  \wedge x\!\!\downarrow)\in last(h)$. From $(x_A^\varphi=\bot)\in last(h)$, we get that $(h, x_A^\varphi)\sim (h,\bot)$, hence  $[h, x_A^\varphi] =[h,\bot]=\bot$. So, to prove that $h(x_A^\varphi)=[h,x_A^\varphi]$, we need to show that $h(x_A^\varphi)=\bot$. For this, we apply Lemma \ref{Unknown-Value} and the assumptions that $x_A^\varphi\!\!\uparrow, \langle K_A\rangle (\varphi  \wedge x\!\!\downarrow) \in last(h)$, to obtain histories $h', h''$ with $h\simA h'\simA h''$,
$\varphi\in last(h')\cap last(h'')$ and $[h',x]\neq [h'',x]$. By the induction hypothesis for both Claims 1 and 2, we get that $h'\models \varphi$ and $h''\models\varphi$ and $h'(x)= [h',x]\neq [h'',x]=h''(x)$. By definition, the existence of histories $h', h''$ with these properties implies that $h(x_A^\varphi)=\uparrow$, as desired.

\emph{Case 2}: $x_A^\varphi\!\!\uparrow\in last(h)$, but $\langle K_A\rangle (\varphi  \wedge x\!\!\downarrow)\not\in last(h)$. Since $\langle K_A\rangle (\varphi  \wedge x\!\!\downarrow)= \langle K_A\rangle (\varphi  \wedge x\neq \uparrow)\in \Sigma$ (because of the closure conditions on $\Sigma$ and the fact that $x_A^\varphi\!\!\uparrow\in last(h)\subseteq \Sigma$), we must have that $K_A(\varphi\to x=\uparrow)\in last(h)$ (since $last(h)$ is a type). By condition \ref{hypo-eq} on types, we have $(x_A^\varphi=\uparrow_A^\varphi)\in last(h)$, which together with condition \ref{undefined} on types gives us $(x_A^\varphi=\uparrow)\in last(h)$, i.e. $(h, x_A^\varphi)\sim (h,\uparrow)$, so $[h, x_A^\varphi]=[h, \uparrow]=\uparrow_\bD$. To prove that $[h, x_A^\varphi]=h(x_A^\varphi)$ as desired, we only need to check now that $h(x_A^\varphi)=\uparrow_\bD$.
For this, let $h' \in H$ s.t. $h\simA h'$ and $h'\models\varphi$, and we need to show $h'(x)=\uparrow_D$.
By the induction hypothesis for Claim 1 and the fact that $K_A(\varphi\to x=\uparrow)\in last(h)$, we have $h\models K_A(\varphi\to x=\uparrow)$, hence $h'\models (\varphi\to x=\uparrow)$, which together with $h'\models\varphi$ gives us $h'\models (x=\uparrow)$. Applying again the induction hypothesis for Claim 1, we have $(x=\uparrow)\in last(h)$, hence $(h,x)\sim (h,\uparrow)$ and so $[h,x]=[h,\uparrow]=\uparrow_D$, as desired.

\emph{Case 3}: $(x_A^\varphi\!\!\downarrow)\in last(h)$. Using Lemma \ref{Know-Value} and the Induction Hypothesis for Claims 1 and 2, we infer the existence of some history $h'\in H$, s.t: $h'\simA h$; $h'\models \varphi$; and $h''(x)= h'(x)=[h, x_A^\varphi]$ for all $h'' \simA h$ s.t. $h''\models \varphi$.
This establishes the existence of a unique value $d:=[h, x_A^\varphi]\in D^\bot$ satisfying the conditions in the first semantic clause for $h(x_A^\varphi)$, so by definition we have $h(x_A^\varphi)=d = [h, x_A^\varphi]$, as desired.
\end{proof}

\medskip

\par\noindent\emph{\textbf{Proof of  Proposition \ref{Satisfiability}}}: If $S$ is a quasi-model for $\varphi_0$, then by applying claim 1 of the Truth Lemma \ref{Truth} to $\varphi_0$ and to the history $h_0:=(\Delta_0)$, we conclude that $h_0\models \varphi_0$ in our history-based data model $\bM$ above.

\section{Completeness and Decidability for the dynamic logic $DLKV$}\label{B}

In this section we prove Theorem \ref{CompDLKV}. We first establish our co-expressivity result: \emph{$DLKV$ and $LKV$ are provably co-expressive}. For this, we need a few preliminary notions and results.

\smallskip\par\noindent\textbf{Expressions and subexpression-complexity} An \emph{expression} $\varepsilon$ of $DLKV$ is any formula $\varphi$, term $x$ or event $e$ of $DLKV$. The \emph{sub-expression complexity order} $>$ is the natural extension to expressions of the usual notion of subformula-complexity; more precisely, $>$ is the least transitive relation on expressions, s.t.: (1) every formula is $>$ all its subformulas and all terms occurring in it; (2) every term is $>$ all its subterms; (3) every event $e=\Phi/\sigma$ is $>$ than all formulas $\varphi\in \Phi$ and all postconditions $\sigma(v)$ (for all $v\in V$); and (4) the formula $[e] \varphi$ is $>e$.

It is easy to see that \emph{$<$ is a well-founded partial order}.

\medskip\par\noindent\textbf{Reducible expressions}
A formula $\theta$ in $DLKV$ is \emph{reducible} if it is provably equivalent to a 'static' formula; i.e., if there exists some formula $\theta'$ in the static fragment $LKV$ s.t.
$\vdash \, \theta \leftrightarrow \theta'$
is provable in $\mathbf{DLKV}$. A term $x\in Var_{DLKV}$ is \emph{reducible} if it is provably equal to a static term; i.e., if there exists some $x'\in Var_{LKV}$ s.t. $\vdash \, x=x'$ is provable in $\mathbf{DLKV}$. Finally, an event $e=\Phi/\sigma$ is \emph{reducible} if all formulas $\varphi\in \Phi$ and all its post-conditions $\sigma(v)$ are reducible (for all $v\in V$).

\smallskip
\par\noindent\textbf{Observation}. \emph{If  $e=\Phi/\sigma$ is a reducible event, then its precondition $pre_e:=\bigwedge\Phi$ is a reducible formula}.

\begin{lemma}\label{conj-red}  The following are provable in $\mathbf{DLKV}$:
\begin{enumerate}
\item \emph{(Derived Reduction Axiom for Conjunction)}
$\vdash \,[e] (\varphi\wedge \psi)  \leftrightarrow  \left([e]\varphi \wedge [e]\psi\right)$
\item \emph{(Existential Dynamic Modality)}
 $\vdash \langle e\rangle \theta \leftrightarrow \left(pre_e\wedge [e]\theta \right)$
 \end{enumerate}
\end{lemma}

\begin{proof} The first is a standard consequence of $[e]$-Necessitation and $[e]$-Distributivity, that holds for all normal modalities. The second follows immediately from the Partial Functionality reduction axiom, together with the definition of $\langle e\rangle \theta$ as an abbreviation for $\neg [e]\neg\theta$.
\end{proof}

\begin{lemma}\label{re} \emph{(Replacement of Equivalents)}
Suppose that $\vdash \varphi \leftrightarrow \varphi'$, $\vdash \theta \leftrightarrow \theta'$, $\vdash x=x'$, $\vdash y=y'$ and $\vdash x_i =x'_i$ (for all $i=1,n$) are provable in $\mathbf{DLKV}$. Then the following are also provable in $\mathbf{DLKV}$:
\begin{enumerate}
\item $\vdash x|_\varphi y = x'|_{\varphi'} y'$
\item $\vdash F(x_1, \ldots, x_n)= F(x'_1, \ldots, x'_n)$
\item $\vdash x_A^\varphi = (x')_A^{\varphi'}$
\item $\vdash e(x)=e(x')$
\item $\vdash Px_1 \ldots x_n \leftrightarrow Px'_1 \ldots x'_n$
\item $\vdash \neg\varphi \leftrightarrow \neg \varphi'$
\item $\vdash (\varphi \wedge \theta) \leftrightarrow (\varphi' \wedge \theta')$
\item $\vdash K_A \varphi \leftrightarrow K_A\varphi'$
\item $\vdash [e]\varphi \leftrightarrow [e]\varphi'$
\item $\vdash C_\gA^\theta \varphi \leftrightarrow C_\gA^{\theta'}\varphi'$
    \end{enumerate}
\end{lemma}

We can now prove a preliminary ``one-step reduction" result:

\begin{lemma}\label{one-step reduction} Let $e$ be any reducible event.
Then, for every 'static' formula $\theta$ in $LKV$ and every 'static' term $x\in Var_{LKV}$, the formula $[e] \theta$ and the term $e(x)$ are also reducible.
\end{lemma}

\begin{proof} Since $e=\Phi/\sigma$ is reducible, its precondition $pre_e=\bigwedge \Phi$ is also reducible. Let us fix some static formula $\rho$ that is provably equivalent to $pre_e$. Throughout the proof, we'll make liberal use of \emph{Replacement of Equivalents} Lemma \ref{re}, without explicitly mentioning it.

We now prove both claims simultaneously, by induction on sub-expression complexity:

\smallskip

For $\theta:=Px_1\ldots x_n$: by the induction hypothesis, for all $i=1,n$ we have $\vdash e(x_i)=x'_i$ for some static terms $x'_1, \ldots, x'_n\in Var_{LKV}$. Using also the Indiscernability axiom and the Atomic Change reduction axiom, $[e] \theta$ is provably equivalent to $pre_e \to P x'_1\ldots x'_n$, and hence also to $\rho\to P x'_1\ldots x'_n$.

The case $\theta:=\neg\varphi$: by the induction hypothesis, $\varphi$ is reducible, so there exists some static formula $\varphi_e$ s.t. $\vdash [e] \varphi \leftrightarrow \varphi_e$ is provable. This, together with our Lemma's assumption and the Partial Functionality reduction axiom, gives us $\vdash [e] \theta \leftrightarrow (\rho\to \neg \varphi_e)$.

The case $\theta \, :=\, \varphi\wedge \theta$ is similar: by induction, there exist static formulas $\varphi_e$ and $\theta_e$ s.t. $\vdash [e] \varphi \leftrightarrow \varphi_e$  and $\vdash [e] \theta \leftrightarrow \theta_e$ are provable. Using now Lemma \ref{conj-red}.1
(the Derived Reduction Axiom for Conjunction) we obtain $\vdash [e] \theta \leftrightarrow (\varphi_e\wedge \theta_e)$.

For $\theta:= K_A\varphi$: by induction, there exists some static formula $\varphi_e$ s.t. $\vdash [e] \rangle \varphi \leftrightarrow \varphi_e$ is provable. Putting this together with the Knowledge Update reduction axiom, we get $\vdash [e] \theta \leftrightarrow (\rho \to K_{e(A)} \varphi_e)$.

For $\theta:= C_{\gA}^\theta \varphi$: we use the induction hypothesis to get static formulas $\varphi_e$ and $\theta_e$ s.t. $\vdash [e] \varphi \leftrightarrow \varphi_e$  and $\vdash [e] \theta \leftrightarrow \theta_e$ are provable. Using now the $C_{\gA}^\theta \varphi$-Update reduction axiom, as well as Lemma \ref{conj-red}.2, we obtain $\vdash [e] \theta \leftrightarrow (\rho \to C_{e[\gA]}^{\rho\wedge \theta_e} \varphi_e)$.

Moving to terms: the case $x:=c\in C$ is trivial, since the Preservation of Constants reduction axiom implies that we have $\vdash e(c)\, =\, c|_{pre_e}\!\!\uparrow$, so $e(c)$ is reducible: $\vdash e(c) \, = \, c|_{\rho}\!\!\uparrow$.

For $x:=v$ (basic variable): since $e$ is reducible, there exists some static term $x'$ s.t. $\vdash post_e(v)= x'$. Using this and the Change of Basic Variables reduction axiom, we obtain that $\vdash e(x)\, =\, x'|_{\rho}\!\!\uparrow$, so $e(x)$ is reducible.

For $x:=y|_\varphi z$: by the induction hypothesis, there exist static terms $y', z'$ and static formula $\varphi'$, s.t.
$\vdash e(y)=y'$, $\vdash e(z)=z'$ and $\vdash [e]\varphi \leftrightarrow \varphi'$ are provable. Using these, as well as the Change of Disjunctive Terms reduction axiom and Lemma \ref{conj-red}.2, we obtain $\vdash e(x)\, =\, (y' |_{\rho \wedge \varphi'} z')|_{\rho}\!\!\uparrow$.

For $x:=F(x_1, \ldots, x_n)$: by the induction hypothesis, there exist static terms $x'_1, \ldots, x'_n$ s.t. $\vdash e(x_i)=x'_i$ is provable for all $i\leq n$. Using the Change of Functional Terms reduction axiom, we obtain $\vdash e(x)\, =\, F(x'_1, \ldots, x'_n)|_{\rho}\!\!\uparrow$.

For $x:=y_A^\varphi$: by induction, there exists a static formula $\varphi'$, s.t. $\vdash [e]\varphi \leftrightarrow \varphi'$ is provable. Using also the Change of Hypothetical Values reduction axiom, as well as Lemma \ref{conj-red}.2, we conclude that $\vdash e(x)\, = \, e(y)_{e(A)}^{\rho\wedge \varphi'}|_{\rho}\!\!\uparrow$.
\end{proof}

We are now in the position to establish our full reduction result:

\begin{lemma}\label{reduction} \emph{(Co-expressivity of $DLKV$ and $LKV$)}
All expressions (formulas, terms and events) $\varepsilon$ of $DLKV$ are reducible. As a consequence, $DLKV$ and $LKV$ have the same expressive power.
\end{lemma}

\begin{proof}
We prove this by induction on the subexpression-complexity of the expression $\varepsilon$.

The base cases $\varepsilon:= c\in C$ and $\varepsilon:=v\in V$ are trivial, since these are already static terms.

The cases $\varepsilon:= x|_\varphi y$, $\varepsilon:= F(\ux)$ and $\varepsilon:=x_A^\varphi$ are straightforward: we use the induction hypothesis and parts 1,2 and 3 of Lemma \ref{re}.

For the case $\varepsilon:= e(x)$: by induction, $e$ and $x$ are reducible; so there exists a static term $x'$ s.t. $\vdash x=x'$ is provable. By Lemma \ref{re}.4, we have $\vdash \varepsilon=e(x')$, and by Lemma \ref{one-step reduction} $e(x'$ is reducible (since $e$ is reducible and $x'$ is static) so we have $\vdash e(x')=x''$ for some static term $x''$. Putting these together and using transitivity of equality, we obtain $\vdash \varepsilon = x''$, so $\varepsilon$ is reducible.

The cases $\varepsilon:= P\ux$, $\varepsilon:= \neg \varphi$, $\varepsilon:=\varphi\wedge \theta$, $\varepsilon:=K_A\varphi$ and $\varepsilon:=C_{\gA}^\theta \varphi$ are again straightforward: use the induction hypothesis and parts 5,6,7,8 and 10 of Lemma \ref{re}.

For the case $\varepsilon:=[e]\varphi$: by induction, $e$ and $\varphi$ are reducible; so there exists a static formula $\varphi'$ s.t. $\vdash \varphi \leftrightarrow \varphi'$ is provable. By \ref{re}.9, we have $\vdash \varepsilon \leftrightarrow [e]\varphi'$, and by \ref{one-step reduction} $[e]\varphi'$ is reducible (since $e$ is reducible and $\varphi'$ is static) so we have $\vdash [e]\varphi' \leftrightarrow \varphi''$ for some static formula $\varphi''$. Putting these together and using transitivity of equivalence, we obtain $\vdash \varepsilon \leftrightarrow \varphi''$, hence $\varepsilon$ is reducible.

Finally, the case $\varepsilon:=e=\Phi/\sigma$ is straightforward: since all the formulas $\varphi\in \Phi$ and all postconditions $\sigma(v)$ (with $v\in V$) are $<e$ in the subexpression complexity order, they are all reducible (by the induction hypothesis), and hence $e$ is also reducible (by definition).
\end{proof}


Finally, we can prove Theorem \ref{CompDLKV} (completeness and decidability of $\mathbf{DLKV}$):

\medskip

\par\noindent\emph{\textbf{Proof of Theorem \ref{CompDLKV}}}: For completeness, let $\theta$ be a valid formula in the language $DLKV$. By Lemma \ref{reduction}, there exists some static formula $\theta'$ in $LKV$ s.t. $\vdash \theta\leftrightarrow \theta'$ is a theorem in the axiomatic system $\mathbf{DLKV}$. By the soundness of $\mathbf{DLKV}$, $\theta'$ is also valid. By
the completeness of the system $\mathbf{LKV}$ (Theorem \ref{CompLKV}), $\vdash \theta'$ is provable in the system $\mathbf{LKV}$ (since $\theta'$ is a static formula), and hence it is also a theorem in the system $\mathbf{DLKV}$. Using this, together with the $\mathbf{DLKV}$-theorem $\vdash \theta\leftrightarrow \theta'$ again, we conclude that $\vdash \theta$ is provable in $\mathbf{DLKV}$, as desired. Decidability follows immediately from the decidability of the static logic $LKV$ (Theorem \ref{CompLKV}) and the co-expressivity of $LKV$ and $DLKV$ (Lemma \ref{reduction}).


\section{Conclusions}


In this paper we axiomatize a decidable but richly expressive logic, that can deal with both ``knowledge that" and ``knowledge what", in a multi-agent setting in which data is modeled using local variables that can be shared, exchanged or modified. Our investigation builds on prior work on knowledge of values and on data-exchange events, as well as more classical work in Epistemic Logic on public announcements, common and distributed knowledge. But our setting has a number of innovative features, such as definite descriptions for hypothetical values, conditional distributed knowledge of values, a conditional version of common distributed knowledge, etc. The proofs also present a number of challenges (due to the presence of the above features).
Due to page limitations, we confined ourselves to semi-public events, but in future work we plan to extend this study to arbitrary data-exchange events (such as secret hacking).





\end{document}